\documentclass[12pt,reqno]{article}

\usepackage{amsmath}
\usepackage{amssymb}
\usepackage{amsthm}
\usepackage{bm}
\usepackage[bf,footnotesize]{caption}
\usepackage{color}
\usepackage{comment}
\usepackage{eucal}
\usepackage{fullpage}
\usepackage{graphicx}
\usepackage{hhline}
\usepackage{lineno}
\usepackage{longtable}
\usepackage{mathpazo}
\usepackage{mathrsfs}
\usepackage{mathtools}
\usepackage[numbers,sort&compress,numbers]{natbib}
\usepackage{setspace}
\usepackage{subfig}
\usepackage{tcolorbox}
\usepackage{times}
\usepackage{xr}
\usepackage[all,cmtip]{xy}

\newcommand*\patchAmsMathEnvironmentForLineno[1]{%
	\expandafter\let\csname old#1\expandafter\endcsname\csname #1\endcsname
	\expandafter\let\csname oldend#1\expandafter\endcsname\csname end#1\endcsname
	\renewenvironment{#1}%
	{\linenomath\csname old#1\endcsname}%
	{\csname oldend#1\endcsname\endlinenomath}}% 
	\newcommand*\patchBothAmsMathEnvironmentsForLineno[1]{%
	\patchAmsMathEnvironmentForLineno{#1}%
	\patchAmsMathEnvironmentForLineno{#1*}}%
	\AtBeginDocument{%
	\patchBothAmsMathEnvironmentsForLineno{equation}%
	\patchBothAmsMathEnvironmentsForLineno{align}%
	\patchBothAmsMathEnvironmentsForLineno{flalign}%
	\patchBothAmsMathEnvironmentsForLineno{alignat}%
	\patchBothAmsMathEnvironmentsForLineno{gather}%
	\patchBothAmsMathEnvironmentsForLineno{multline}%
}

\newcommand{\vx}{\mathbf{x}}
\newcommand{\vy}{\mathbf{y}}
\newcommand{\vX}{\mathbf{X}}
\newcommand{\MSS}{\mathrm{MSS}}

\newcommand{\supp}{Supporting Information}

\theoremstyle{definition}
\newtheorem{corollary}{Corollary}

\newtheorem{example}{Example}
\newtheorem{lemma}{Lemma}
\newtheorem{proposition}{Proposition}
\newtheorem{remark}{Remark}
\newtheorem{theorem}{Theorem}

\newcommand{\eq}[1]{\textbf{Eq.~\ref{eq:#1}}}
\newcommand{\fig}[1]{\textbf{Fig.~\ref{fig:#1}}}
\newcommand{\tab}[1]{\textbf{Tab.~\ref{tab:#1}}}
\newcommand{\lem}[1]{Lemma~\ref{lem:#1}}
\newcommand{\prop}[1]{Proposition~\ref{prop:#1}}
\newcommand{\thm}[1]{Theorem~\ref{thm:#1}}

\newcommand{\pageof}[1]{Page~\pageref{symbol:#1}}
\newcommand{\pageofeq}[1]{Page~\pageref{eq:#1}}

\title{\begin{center} \bfseries\singlespacing
Evaluating the structure-coefficient theorem of evolutionary game theory
\end{center}}
\author{\parbox[c]{16cm}{\onehalfspacing \normalsize \centering ~\\[-0.4cm] Alex McAvoy$^{1,2,3}$ \& John Wakeley$^{1}$ \\ \quad\\ \footnotesize
		$^{1}$Department of Organismic and Evolutionary Biology, Harvard University, Cambridge, MA~02138 \\
		$^{2}$Department of Mathematics, University of Pennsylvania, Philadelphia, PA~19104 \\
		$^{3}$Center for Mathematical Biology, University of Pennsylvania, Philadelphia, PA~19104 \\[0.2cm]}
\date{}
}

\begin{document}

\allowdisplaybreaks

\maketitle

\begin{abstract}
In order to accommodate the empirical fact that population structures are rarely simple, modern studies of evolutionary dynamics allow for complicated and highly-heterogeneous spatial structures. As a result, one of the most difficult obstacles lies in making analytical deductions, either qualitative or quantitative, about the long-term outcomes of evolution. The ``structure-coefficient'' theorem is a well-known approach to this problem for mutation-selection processes under weak selection, but a general method of evaluating the terms it comprises is lacking. Here, we provide such a method for populations of fixed (but arbitrary) size and structure, using easily interpretable demographic measures. This method encompasses a large family of evolutionary update mechanisms and extends the theorem to allow for asymmetric contests to provide a better understanding of the mutation-selection balance under more realistic circumstances. We apply the method to study social goods produced and distributed among individuals in spatially-heterogeneous populations, where asymmetric interactions emerge naturally and the outcome of selection varies dramatically depending on the nature of the social good, the spatial topology, and frequency with which mutations arise.
\end{abstract}

\section{Introduction}
Population structure is nearly ubiquitous in nature and occurs on many levels. Microbes exhibit strong spatial heterogeneity in response to variation in temperature, pH, and vegetation cover in terrestrial soils \citep{franklin:FEMSME:2003} and to harsh environmental gradients and patchy habitats at deep-sea hydrothermal vents \citep{harmsen:AEM:1997}. On a very different scale, global human social networks are known to consist of many highly-connected clusters, often held together by dense cores \citep{mislove:IMC:2007}, which can in turn affect both public opinion \citep{lee:JC:2014} and disease transmission \citep{jo:SR:2021,nande:PLOSCB:2021}. Evolutionary outcomes depend on population structure because it constrains social interactions and the dispersal of offspring.

Fully a century ago, \citet{handh:1921} argued based on observations of spatial heterogeneity in land snails that variation due to chance in small local populations is an important factor of evolution. Early theoretical work by \citet{wright:GEN:1931,wright:GEN:1943} and \citet{malecot:1948} on the distribution of gene frequencies and genetic identity in regularly subdivided populations stimulated later research in population and evolutionary genetics, focused largely on symmetric population structures with discrete, well-mixed subpopulations connected by migration. These include the island model \citep{wright:GEN:1931}, with symmetric migration between every pair of subpopulations, and the one-dimensional and two-dimensional stepping-stone models \citep{kimura:GEN:1964,maruyama:AHG:1970,maruyama:AHG:1971}, with migration only between neighboring subpopulations on a lattice. These same models have been used to study social evolution; e.g.\ see \citet{nakamaru:JTB:1997}. 

Decades of subsequent research clarified the empirical facts of population structure as well as its importance in theoretical models of evolution \citep{slatkin:ARES:1985,slatkin:Science:1987}. Increased relatedness and thus greater genetic identity within local populations is the hallmark of population structure. The cause is local ``genetic drift'' (hereafter just drift) due to the vagaries of reproduction when the number of individuals in a given region is not large and migration is restricted. Frequencies of traits vary from one location to the next, which in turn affects the outcome of selection because it alters the rate of same-type interactions. Using diffusion models, for example, it has been shown that high local identity between individuals can change the effective population size and selection coefficient \citep{cherry:GEN:2003}. These effects can dramatically alter the course of evolution (e.g. through increasing the fixation probability of recessive advantageous mutations \citep{roze:GEN:2003}). Ideas of non-uniform interaction have also been crucial in evolutionary game theory, even in models of otherwise well-mixed, infinite populations \citep{hamilton:JTB:1964a,hamilton:JTB:1964b,grafen:AB:1979,hines:JTB:1979,tao:JTB:2002,taylor:TPB:2006}. That interaction structure and reproductive structure may work together non-independently has been further emphasized by laboratory experiments \citep{wade:EVOL:1980,kerr:Nature:2002}.

Population structure affects who interacts with whom, but it can also fundamentally change the nature of these interactions. In what follows, we use the example of competition between producers of a good (a social trait) and non-producers. Producers generate a good of benefit $b$\label{symbol:b} and pay a cost, $c$\label{symbol:c}, to do so. While the production of a good in isolation is straightforward, its expression within a structured population can take different forms. One option is for a producer to generate a separate good for each neighbor, each of benefit $b$ (to the neighbor) and cost $c$ (to the producer). As a producer's group size grows, so too does the total good generated by the producer, a quality that has been observed in social grooming among primates \citep{lehmann:AB:2007}. For other kinds of divisible goods, such as food shared among vampire bats \citep{wilkinson:Nature:1984} or the time and effort a researcher allocates across several collaborations \citep{jackson:JET:1996}, a more natural assumption might be that the good is produced once and the benefit, $b$, is divided among neighbors \citep{akcay:NC:2018}. How much any one neighbor receives depends on how many share the benefit, which can vary significantly in a spatially-heterogeneous population. Allotments of social goods in structured populations lead naturally to asymmetric games, in which payoffs to individuals depend on not just the trait of immediate interest (e.g.\ producer versus non-producer) but on other factors as well, including group sizes \citep{maynardsmith:CUP:1982,broom:TF:2013}.

In this work, we consider the distribution of social goods in a general model of a graph-structured population, which can include the discrete-subpopulation models mentioned above but more importantly can depict arbitrarily fine-grained heterogeneity \citep{lieberman:Nature:2005,ohtsuki:PRL:2007,ohtsuki:JTB:2007,broom:PRSA:2008}. We are interested in whether a social trait is favored in evolution, with competition determined by a game that may be asymmetric. With the greater flexibility of an arbitrary graph comes greater complexity, which prohibits comprehensive analysis for all but the smallest populations or those with high degrees of symmetry. This trade-off between model complexity and tractability has prompted a variety of approaches ranging from approximation techniques \citep{ohtsuki:Nature:2006,szabo:PR:2007} to exact results on homogeneous population structures \citep{taylor:Nature:2007,chen:AAP:2013,debarre:NC:2014,mullon:JEB:2014,chen:SR:2016} to the analysis of simulation algorithms \citep{hindersin:SR:2019}. We follow a common fruitful way forward, namely to assume that selection is weak and to focus on the stationary frequency of the social trait, in which case increased local relatedness and its effect on evolution can be captured with relative ease \citep{taylor:JTB:1996,rousset:JEB:2000,nowak:Nature:2004,lessard:JMB:2007,antal:PNAS:2009,durrett:EJP:2014}.

When selection is weak and mutations appear sufficiently infrequently, the problem can be reduced to a consideration of fixation probabilities \citep{rousset:JEB:2000,nowak:Nature:2004,fudenberg:JET:2006,vancleve:TPB:2015}. Recent techniques allow one to compute fixation probabilities in populations with arbitrary spatial structure \citep{allen:Nature:2017,mcavoy:JMB:2021}. Under the assumption that a mutation either fixes or is removed from the population prior to the appearance of another mutant, fixation probabilities allow faithful approximations of the mean frequency of a trait \citep{fudenberg:JET:2006}. This assumption is slightly more restrictive than it first appears, as the amount of time it takes for the population to return to a monomorphic state can depend strongly on the nature of the interactions. For instance, it is known that for some kinds of frequency-dependent interactions, in which coexistence is possible for extended periods of time, the maximum allowable mutation probability decays exponentially in the population size \citep{antal:BMB:2006,wu:JMB:2011}. For these kinds of traits, fixation probabilities are no longer necessarily the most meaningful or relevant measures of evolutionary dynamics.

Aside from the question of whether mutations are rare enough for fixation probabilities to be appropriate measures, we note that the mutation probability need not be small at all. In cultural contexts, mutations can reasonably occur with high frequency \citep{frank:EVOL:1997,tarnita:JTB:2009,tarnita:PNAS:2011,traulsen:PNAS:2009,debarre:JTB:2017,debarre:DGA:2020}. Assuming weak selection but not weak mutation, \citet{tarnita:JTB:2009} elegantly boiled the question of whether a trait is favored down to a single quantity called the structure coefficient. To state their theorem, we use the notation of \citet{debarre:NC:2014} and assume that there are two types in the population: $\textrm{S}$\label{symbol:S} (``social'') and $\textrm{NS}$\label{symbol:NS} (``non-social''). Payoffs for interactions between types are given by the matrix
\begin{align}
\bordermatrix{%
& \textrm{S} & \textrm{NS} \cr
\hfill\textrm{S} &\ a_{11} & \ a_{12} \cr
\textrm{NS} &\ a_{21} & \ a_{22} \cr
}\ . \label{eq:symmetric_payoff_matrix}
\end{align} 
The social trait, $\textrm{S}$, is said to be favored by selection (via the game) if its equilibrium frequency is greater than what it would be under neutrality, i.e. if the first-order effect of selection on the mean frequency of the type is positive \citep{rousset:JEB:2000,nowak:Nature:2004}. As this single matrix applies to all pairs of individuals, this is necessarily a symmetric game. Following \citet{tarnita:JTB:2009}, we assume that mutations occur with fixed probability $u>0$\label{symbol:u} per reproduction event, and that the resulting trait is sampled uniformly at random from the space of possible types. In particular, when there are two traits, $\textrm{S}$ is favored by selection if its equilibrium frequency is above $1/2$. Under some mild conditions on the population structure and the evolutionary update rule, \citet{tarnita:JTB:2009} show that there exists a ``structure coefficient,'' $\sigma$\label{symbol:sigma}, such that weak selection favors $\textrm{S}$ over $\textrm{NS}$ in the mutation-selection equilibrium whenever
\begin{align}
\sigma a_{11} + a_{12} &> a_{21} + \sigma a_{22} . \label{eq:sigma_inequality}
\end{align}

The structure coefficient depends on the mutation probability, population structure, and update rule, but not on the payoffs of the game. Thus, under weak selection and considering only the mean frequency of $\textrm{S}$ at equilibrium, $\sigma$ captures all relevant aspects of population structure. It may be noted that $\sigma$ appears in \eq{sigma_inequality} as a weight on the payoffs of same-type interactions. If one models the effects of selection on fecundity and survival separately, using different payoff matrices for each, then the structure-coefficient theorem involves two weights, one for each step of the life cycle \citep{debarre:NC:2014}. These weights still retain the property of applying to same-type interactions, and when the two payoff matrices coincide (or one is zero), a condition involving only one structure coefficient is recovered. Our extension of the structure-coefficient theorem includes two coefficients, though we do not distinguish between different steps of a life cycle using different payoff matrices.

\citet{tarnita:JTB:2009} proved the existence of $\sigma$ but did not give a general method for computing it. One of our goals here is to provide a recipe for calculating $\sigma$ in terms of simple demographic quantities. However, knowing how to calculate $\sigma$ still does not allow us to study the production of divisible social goods. To see why, we can return to the two kinds of social goods considered previously. Following the terminology of \citet{mcavoy:NHB:2020}, we refer to goods in which each neighbor gets $b$ at a cost of $c$ as ``pp-goods'' since both the total benefit and the total cost are proportional (``p'') to the number of interaction partners. The other kind of good, in which a single benefit, $b$, is divided among all neighbors, is termed an ``ff-good'' to indicate that the total benefits and costs are both fixed (``f''), independent of the number of neighbors. In both cases the interaction between any pair of neighboring individuals is a kind of ``donation game'' \citep{sigmund:PUP:2010} but the combined outcome of all such pairwise games across the population differs for pp-goods versus ff-goods.

The case of pp-goods falls within the scope of the model of \citet{tarnita:JTB:2009} since the benefits and costs are the same for every pair of neighbors in the population. In terms of \eq{symmetric_payoff_matrix}, pp-goods satisfy $a_{11}=b-c$, $a_{12}=-c$, $a_{21}=b$, and $a_{22}=0$. Here games are symmetric and population structure only affects $\sigma$. In contrast, with ff-goods population structure usually results in asymmetric games and cannot be summarized by a single coefficient. The payoff to an individual of type $x$ at location $i$ in the game against an individual of type $y$ at location $j$ depends on their numbers of neighbors. We write $a_{xy}^{ij}$\label{symbol:asymmetric_payoff_matrix} to denote the dependence on their locations and suppose that $w_{i}$\label{symbol:degree} and $w_{j}$ are their numbers of neighbors. For brevity in much of what follows, we refer to the individuals at locations $i$ and $j$ simply as ``$i$'' and ``$j$.'' If $i$ is a producer, then $j$ receives a benefit $b/w_{i}$, and $i$ pays a cost $c/w_{i}$ for having $j$ as a neighbor. Similarly, if $j$ is a producer, then $j$ pays $c/w_{j}$ to provide a benefit of $b/w_{j}$ to $i$. Thus, for ff-goods, $a_{11}^{ij}=\left(b/w_{j}-c/w_{i}\right)$, $a_{12}^{ij}=-c/w_{i}$, $a_{21}^{ij}=b/w_{j}$, and $a_{22}^{ij}=0$. Games involving ff-goods in structured populations are symmetric only when every individual has the same number of neighbors such that $a_{xy}^{ij}$ is independent of $i$ and $j$.

Here, we provide a method of analyzing mutation-selection dynamics in populations with finite size, $N$\label{symbol:N}, and fixed (but arbitrary) structure. Specifically, we obtain an extension of the structure-coefficient theorem to games written in terms of $a_{xy}^{ij}$. We prove a result analogous to \eq{sigma_inequality} but which involves more than a single structure coefficient. This allows us to quantify how much a strategy is favored or disfavored, not just whether its mean frequency is above or below the value expected under neutral drift. We provide a general method of computing these structure coefficients in terms of simple demographic quantities, and give a complexity bound on the calculations in terms of the population size. We also consider games with any finite number of strategies, providing a method for calculating the coefficients appearing in the generalization of the structure-coefficient theorem due to \citet{tarnita:PNAS:2011}. 

We use this method to investigate pp-goods versus ff-goods in graph-structured populations, in order to understand how the differences between these goods may be amplified by spatial heterogeneity as well as how the rate of mutation affects the extent to which they are each favored (or disfavored) by selection. The results show strong effects of both population structure and mutation probability, in terms of which traits are favored and by how much. Trivially, as the mutation probability approaches one, selection becomes ineffective, independent of the type of good. But the approach to this can be non-monotonic. The extent to which a trait is favored, and even whether it is favored or disfavored, depends on the mutation probability. Further, the relative advantage or disadvantage of pp-goods versus ff-goods depends on mutation. Finally, we use simulations to assess the sensitivity of these results beyond the limits of weak selection, and we find them to be qualitatively robust.

\section{Theory and Results}
We begin with an overview of the theory developed in detail in {\supp}. We discuss the key parameters involved in assessing the effects selection on the mean frequencies of strategies or traits, and we describe how these quantities are used to compute structure coefficients. \tab{notationtable} provides a glossary of all the notation we use.

\begin{longtable}{cp{9cm}l}
\caption{Glossary of Notation}
\label{tab:notationtable}\\
\hline
Symbol & Description & Defined \\
\hline
\endhead
\hline \endfoot
$b$ & Benefit of a social good & \pageof{b} \\
$c$ & Cost of a social good & \pageof{c} \\
$\textrm{S}$ & The mutant (social) trait & \pageof{S} \\
$\textrm{NS}$ & The resident (non-social) trait & \pageof{NS} \\
$a_{xy}$ & Payoff to $x$ against $y$ in a symmetric game & \pageofeq{symmetric_payoff_matrix} \\
$u$ & Probability of a mutation on reproduction & \pageof{u} \\
$\sigma$ & Structure coefficient of \citet{tarnita:JTB:2009} & \pageof{sigma} \\
$w_{i}$ & Number of neighbors (degree) of $i$ & \pageof{degree} \\
$a_{xy}^{ij}$ & Payoff to $x$ at $i$ against $y$ at $j$ in an asymmetric game & \pageof{asymmetric_payoff_matrix} \\
$N$ & Population size & \pageof{N} \\
$\delta$ & Selection intensity & \pageof{delta} \\
$\left< x_{\textrm{S}}\right>$ & Equilibrium mean frequency of $\textrm{S}$ & \pageof{mean_frequency} \\
$K_{1}^{k\ell},K_{2}^{k\ell}$ & Structure coefficients for asymmetric games & \pageofeq{structure_coefficients} \\
$R$ & Set of individuals to be replaced & \pageof{R} \\
$\alpha$ & Offspring-to-parent map & \pageof{alpha} \\
$n$ & Number of traits & \pageof{n} \\
$\vx$ (resp. $x_{i}$)& State of the population (resp. of $i$) & \pageof{vx} \\
$p_{\left(R,\alpha\right)}\left(\vx\right)$ & Probability of replacement event $\left(R,\alpha\right)$ in state $\vx$ & \pageof{pRalpha} \\
$e_{ij}\left(\vx\right)$ & Probability that $i$ replaces $j$ in state $\vx$ & \pageof{eij} \\
$d_{i}\left(\vx\right)$ & Death probability of $i$ in state $\vx$ & \pageof{di} \\
$A_{ij}$ & Probability of moving from $i$ to $j$ in the ancestral walk & \pageof{Aij} \\
$z_{i}$ & Equilibrium probability that an ancestral line is at $i$ & \pageof{zi} \\
$\pi_{i}$ & Reproductive value of $i$ & \pageof{pi} \\
$z_{i}^{\textrm{mut}}$ & Mean probability that first ancestral mutation is at $i$ & \pageofeq{mui} \\
$\pi_{i}^{\textrm{mut}}$ & Mutation-weighted reproductive value of $i$ & \pageof{pi_mut} \\
$\phi_{ij}$ & Probability that $i$ and $j$ are identical by state & \pageof{phi} \\
$U_{i}\left(\vx\right)$ & Payoff to $i$ in state $\vx$ & \pageof{U_i} \\
$\mathbf{F}\left(\vx\right)$ (resp. $F_{i}\left(\vx\right)$) & Fecundity of the population (resp. of $i$) in state $\vx$ & \pageof{F_i} \\
$m_{k}^{ij}$ & Marginal effect of $k$'s fecundity on $i$ replacing $j$ & \pageof{mkij} \\
$p_{\left(x_{i},x_{j}\right)}^{ij}$ & Neutral probability that $i$ is $x_{i}$ and $j$ is $x_{j}$ & \pageof{neutral_pair} \\
$p_{\left(x_{i},x_{j},x_{k}\right)}^{ijk}$ & Neutral probability that $i$ is $x_{i}$, $j$ is $x_{j}$, and $k$ is $x_{k}$ & \pageof{neutral_triplet} \\
$\Omega_{ij}$ & Interaction structure of the population & \pageof{omega} \\
$w_{ij}$ & Adjacency matrix of the population & \pageof{w_ij} \\
\hline
${}^{\circ}$ & Indicates neutral drift ($\delta =0$) & $\--$
\end{longtable}

Following \citep{rousset:JEB:2000,nowak:Nature:2004} we consider the first-order effect of selection on the mean frequency of a trait when the population is at stationarity. We use the general model of population structure of \citet{allen:JMB:2019} in which a scalar parameter $\delta$\label{symbol:delta} captures the intensity or strength of selection. By stationarity, we mean statistical equilibrium with respect to the processes of mutation, selection via a game in which traits are strategies, and reproduction in a finite population. If $\left< x_{\textrm{S}}\right>$\label{symbol:mean_frequency} is the mean frequency of the social trait $\textrm{S}$ in this stationary distribution, then $\textrm{S}$ is said to be favored when $\left< x_{\textrm{S}}\right>$ is greater than the corresponding prediction under neutrality. This is the framework of the structure-coefficient theorem, which follows from considering $\frac{d}{d\delta}\left< x_{\textrm{S}}\right>$ evaluated at $\delta =0$.

We take the approach in \citet{mcavoy:NHB:2020} and consider the effects of a deviation from neutrality on the fecundities of individuals ($\mathbf{F}$), as well as the effects of the resulting deviations in fecundities on the frequency of the trait. Fecundities depend on the game payoffs and the structure of the population. Thus, we compute the first-order effects of selection using the relationship 
\begin{align}\label{eq:chain}
\frac{d\left< x_{\textrm{S}}\right>}{d\delta} &= \frac{d\left< x_{\textrm{S}}\right>}{d\mathbf{F}} \frac{d\mathbf{F}}{d\delta} .
\end{align} 

Our solution involves three types of measures of population heterogeneity, which in turn depend on fundamental demographic parameters. The first measure, $\pi^{\textrm{mut}}$, is a type of reproductive value that summarizes patterns of identity by descent across the population. The second, $\phi$, is a measure of identity by state, for all pairs of individuals in the population, which determines expected payoffs to individuals. The third, $m_{k}^{ji}$, is the effect of a deviation from neutral fecundity in individual $k$ on the probability that $j$ replaces $i$ in one time step. While the first two measures are evaluated under neutrality ($\delta =0$), the third measure links these quantities to the effects of selection. Note again that, for brevity, by ``$j$ replaces $i$'' we mean the event that the individual at location $j$ replaces the individual at location $i$, and similarly that $i$, $j$, and $k$ are technically locations within the population, not individuals.

Our main result, which is derived in Supplementary Information, says that
\begin{align}
\left< x_{\textrm{S}} \right> &= \frac{1}{2} + \delta\frac{1}{2}\sum_{k,\ell =1}^{N}\left(\substack{K_{1}^{k\ell}\left(a_{11}^{k\ell}-a_{22}^{k\ell}\right) \\ +K_{2}^{k\ell}\left(a_{12}^{k\ell}-a_{21}^{k\ell}\right)}\right) + O\left(\delta^{2}\right) , \label{eq:xA_expansion}
\end{align}
where
\begin{subequations}\label{eq:structure_coefficients}
\begin{align}
K_{1}^{k\ell} &= \frac{1}{2u} \sum_{i,j=1}^{N} \pi_{i}^{\textrm{mut}} m_{k}^{ji} \left( \substack{- \left(\phi_{ik}+\phi_{i\ell}\right) \\ + \left(1-u\right)\left(\phi_{jk}+\phi_{j\ell}\right) + u} \right) ; \\
K_{2}^{k\ell} &= \frac{1}{2u} \sum_{i,j=1}^{N} \pi_{i}^{\textrm{mut}} m_{k}^{ji} \left( \substack{- \left(\phi_{ik}-\phi_{i\ell}\right) \\ + \left(1-u\right)\left(\phi_{jk}-\phi_{j\ell}\right)} \right) .
\end{align}
\end{subequations}
Again we assume that $u>0$, but we note that \eq{structure_coefficients} does not diverge as $u\rightarrow 0$. Both $- \left(\phi_{ik}+\phi_{i\ell}\right) + \left(1-u\right)\left(\phi_{jk}+\phi_{j\ell}\right) + u$ and $-\left(\phi_{ik}-\phi_{i\ell}\right) + \left(1-u\right)\left(\phi_{jk}-\phi_{j\ell}\right)$ in \eq{structure_coefficients} vanish in the limit $u\rightarrow 0$ because all individuals must eventually be identical by state in the absence of mutation. These terms divided by $u$ then each have finite limits as $u\rightarrow 0$. We refer the reader to Refs.~\citep{allen:JMB:2019,mcavoy:JMB:2021} for discussions of the rare-mutation limit in this class of models.

\subsection{Modeling population dynamics}
We consider a discrete-time population model with arbitrary spatial structure. The population evolves through a sequence of replacement events, which consist of pairs, $\left(R,\alpha\right)$, where $R\subseteq\left\{1,\dots ,N\right\}$\label{symbol:R} is the set of individuals to be replaced and $\alpha :R\rightarrow\left\{1,\dots ,N\right\}$\label{symbol:alpha} is the parentage map \citep{allen:JMB:2014,allen:JMB:2019}. In such an event, the individual at location $i\in R$ is replaced by the offspring of the individual at location $\alpha\left(i\right)$. If $i\notin R$, then $i$ simply lives through the time step. We assume that there are no ``empty'' sites.

Each individual has one of $n$\label{symbol:n} traits. We use $\vx$\label{symbol:vx} to denote the current state of the population, which is a list of the trait values of all $N$ individuals. For example, with just $n=2$ possible trait values, $\textrm{S}$ and $\textrm{NS}$, we would have $x_{i}\in\left\{\textrm{S},\textrm{NS}\right\}$ for $i\in\left\{1,\dots,N\right\}$. Here, we consider $n=2$, but our results in {\supp} allow for $n>2$. We use $p_{\left(R,\alpha\right)}\left(\vx\right)$\label{symbol:pRalpha} to denote the probability that replacement event $\left(R,\alpha\right)$ occurs. This ``replacement rule'' \citep{allen:JMB:2014,allen:JMB:2019} depends on the state of the population because individual fecundities are affected by game payoffs. 

We allow symmetric, parent-independent mutation with probability $u$ per individual offspring in a replacement event. That is, with probability $1-u$, the offspring of $\alpha\left(i\right)$ that replaces $i$ has the same strategy as its parent. With probability $u$, the offspring mutates and acquires a strategy uniformly-at-random, for example from $\left\{\textrm{S},\textrm{NS}\right\}$ when $n=2$. Thus, offspring and parent might still be identical by state after a mutation occurs, but they will not be identical by descent. We assume that the mutation probability is positive and is the same for all individuals.

If $u>0$ and a mild assumption on the replacement rule $p_{\left(R,\alpha\right)}\left(\vx\right)$ holds --- namely that at least one individual can propagate its offspring throughout the entire population (see {\supp}, \S\ref{sec:replacement_rules}) --- then this process has a unique stationary distribution with $0< \left< x_{\textrm{S}}\right> <1$. We study this distribution under the assumption of weak selection inherent to the structure-coefficient theorem, with selection intensity $\delta\geqslant 0$ determining the importance of the game to fecundity. Note that $p_{\left(R,\alpha\right)}\left(\vx\right)$ does not depend on $\vx$ when $\delta =0$, because trait values do not affect fecundity under neutrality. We use $p_{\left(R,\alpha\right)}^{\circ}$ to denote the corresponding neutral replacement rule.

The fundamental demographic parameters of individuals are marginal properties of $p_{\left(R,\alpha\right)}\left(\vx\right)$. The first and most important is the probability that $i$ replaces $j$ in one time step,
\begin{align}
e_{ij}\left(\vx\right)\label{symbol:eij} &\coloneqq \sum_{\substack{\left(R,\alpha\right) \\ j\in R,\ \alpha\left(j\right) =i}} p_{\left(R,\alpha\right)}\left(\vx\right) .
\end{align}
Note that this replacement involves the death of the previous individual, $j$. Other demographic quantities are functions of the $e_{ij}\left(\vx\right)$. The second demographic parameter we make extensive use of here is the probability that $i$ dies in one time step, $d_{i}\left(\vx\right)\coloneqq\sum_{j=1}^{N}e_{ji}\left(\vx\right)$\label{symbol:di}. Weak selection causes deviations from the neutral values of these two quantities. We call these neutral values $e_{ij}^{\circ}$ and $d_{i}^{\circ}$ and note they are independent of the population state, $\vx$. Deviations in $e_{ij}$ (and all other demographic quantities) due to selection via the game do depend on $\vx$, but only through changes in the fecundities ($\mathbf{F}$) as in \eq{chain}. We use $\mathbf{F}^{\circ}$ to denote fecundities under neutral drift.

In the following subsections we discuss the three types of measures of population heterogeneity introduced briefly in \eq{structure_coefficients}, which derive from the analysis in {\supp}.

\subsubsection{Reproductive value and identity by descent}
The coefficient $\pi_{i}^{\textrm{mut}}$ in \eq{structure_coefficients} is a mutation-weighted reproductive value, which summarizes patterns of identity by descent between individuals and their direct ancestors at location $i$. To illustrate what $\pi_{i}^{\textrm{mut}}$ represents here, we first recall the classical notion of reproductive value.

Under neutral drift, probabilities of birth, death, and replacement still vary among individuals due to their locations in the population structure. Reproductive values measure the contributions of individuals to distant future generations or, equivalently, the contributions of a long-ago ancestral individuals to the present generation \citep{fisher:OUP:1930,price:AHG:1972,taylor:AN:1990,derrida:PA:2000,engen:AN:2009,barton:G:2011,maciejewski:JTB:2014a}. If $\pi_{i}$ is the reproductive value of $i$, then with probability $1-d_{i}^{\circ}$, individual $i$ survives and brings its value $\pi_{i}$ forth into the next time step; with probability $e_{ij}^{\circ}$, individual $i$ replaces $j$ and, in doing so, acquires the value $\pi_{j}$. At stationarity, the reproductive values satisfy the balance equation $\pi_{i}=\left(1-d_{i}^{\circ}\right)\pi_{i}+\sum_{j=1}^{N}e_{ij}^{\circ}\pi_{j}$.

The solution to this recurrence is unique up to the value of $\sum_{i=1}^{N}\pi_{i}$. A standard normalization is $\sum_{i=1}^{N}\pi_{i}=1$. However, this is not the only relevant normalization, and here it will be useful to work with something different. Reproductive values are closely related to the stationary distribution of a random walk on the population. Let $A_{ij}\coloneqq e_{ji}^{\circ}/d_{i}^{\circ}$\label{symbol:Aij} be the probability that the parent of $i$ was located at $j$, conditioned on the death of $i$. Thus, $j$ replacing $i$ forward in time equates to $i$ moving to $j$ backward in time. This ancestral random walk tracks only parentage and ignores lifespans \citep{allen:JMB:2019}. It has a unique stationary distribution, $z$, which satisfies $z_{i}=\sum_{j=1}^{N}z_{j}A_{ji}$\label{symbol:zi} and represents the equilibrium probability of an ancestral line being at $i$ after many lifetimes. Weighting this probability by the expected lifespan at $i$ gives the reproductive value, $\pi_{i}=z_{i}/d_{i}^{\circ}$\label{symbol:pi}, of $i$.

The distinction between $\pi_{i}$ and $z_{i}$ aids in the interpretation of $\pi_{i}^{\textrm{mut}}$ and its relationship to patterns of identity by descent. Consider the ancestral random walk defined by the matrix $A$. Starting with an individual at location $j$, each step backward in time is associated to a mutation with probability $u$ and to a faithful transmission of genetic material with probability $1-u$. We consider the endpoint, $i$, to be the location of the most ancient individual that is identical by descent with $j$. Thus, the first mutation in the ancestral line from $j$ occurs between $i$ and its parent, $\alpha\left(i\right)$. Since a mutation must eventually arise in this walk, the endpoint defines a distribution over $\left\{1,\dots ,N\right\}$.

For the first mutation to occur at $i$ after $t$ steps in the past, the first $t$ reproduction events in this lineage must be mutation-free, which happens with probability $\left(1-u\right)^{t}$. The probability that the lineage is at $i$ is then $\left(A^{t}\right)_{ji}$. Finally, the parent of $i$ must produce a mutated offspring, which happens with probability $u$. Summing these probabilities over all $t$ gives a total probability of
\begin{align}
\sum_{t=0}^{\infty} u\left(1-u\right)^{t} \left(A^{t}\right)_{ji} &= u\left(\left(I-\left(1-u\right) A\right)^{-1}\right)_{ji} \label{eq:Dji}
\end{align}
that the first mutation happens at $i$, given a starting position of $j$. Because $\left< x_{\textrm{S}}\right>$ is the mean frequency of the trait taken uniformly over the population, our calculations depend on the average
\begin{align}
z_{i}^{\textrm{mut}} &\coloneqq \frac{1}{N}\sum_{j=1}^{N} u \left(\left(I-\left(1-u\right) A\right)^{-1}\right)_{ji} . \label{eq:mui}
\end{align}
$z_{i}^{\textrm{mut}}$ is the analogue of $z_{i}$ but weighted by identity by descent (i.e. with a time scale 
dependent on the mutation probability). In the limit $u\to 0$, $z_{i}^{\textrm{mut}}$ converges to $z_{i}$. At the other extreme, as $u\to 1$, $z_{i}^{\textrm{mut}}$ becomes uniform on $\left\{1,\dots ,N\right\}$ owing to the fact that the mutation will then occur when $t=0$, so \eq{Dji} will tend to one for $i=j$ and zero for $i\neq j$. As a distribution, $\sum_{i=1}^{N}z_{i}^{\textrm{mut}}=1$ for any $u$.

The mutation-weighted reproductive value $\pi_{i}^{\textrm{mut}}$ in \eq{structure_coefficients} is defined by $\pi_{i}^{\textrm{mut}}\coloneqq z_{i}^{\textrm{mut}}/d_{i}^{\circ}$\label{symbol:pi_mut}. Analogous to the expression $\pi_{i}=z_{i}/d_{i}^{\circ}$, mutation-weighted reproductive value captures both the contribution of $i$ to the total population, through the factor $z_{i}^{\textrm{mut}}$, and the time scale of reproduction, through the factor $1/d_{i}^{\circ}$. In the limits $u\to 0$ and $u\to 1$, $\pi_{i}^{\textrm{mut}}$ converges to $\pi_{i}$ and $1/\left(Nd_{i}^{\circ}\right)$, respectively. Further details about how $\pi^{\textrm{mut}}$ arises may be found in {\supp}, \S\ref{sec:mean_change}.

\subsubsection{Identity by state}
When selection is weak, neutral probabilities of identity by state play a key role in outcome of frequency-dependent selection \citep{rousset:JEB:2000,taylor:BMB:2004}. Let $\phi_{ij}$\label{symbol:phi} denote the stationary probability that two individuals at locations $i$ and $j$ are identical by state (have the same type or strategy) under neutral drift. Of course, $\phi_{ii}=1$ for $i=1,\dots ,N$. The other $\phi_{ij}$ for $i \neq j$ can be computed using standard techniques, recursively over a single time step. We present the details in {\supp}, \S\ref{sec:neutral_drift}, where we account for the possible replacement of both $i$ and $j$ and for the possible occurrence of zero, one, or two mutations in \eq{phi_ij_recurrence}. Solving for these $N\left(N-1\right) /2$ probabilities, $\phi_{ij}$ for $i \neq j$, is the main computational burden of our method. In our application to producers of social goods below, we discuss a simple intuitive example of this approach to identities by state.

\subsubsection{Marginal effects of selection on replacement}
Up to this point, our focus has been on neutrality ($\delta = 0$). Since we are interested in the perturbative effects of selection on the neutral process, we also need to understand how changes in individual fecundity affect the transmission of traits. In evolutionary games, the influence of the state of the population on replacement can be modeled using a fecundity function, which assigns to each state a positive quantity representing an individual's competitive abilities (including, but not limited to, the propensity to reproduce). If $U_{i}\left(\vx\right)$\label{symbol:U_i} is the payoff of individual $i$ in state $\vx$, then a typical conversion of payoff to fecundity is of the form $F_{i}\left(\vx\right) =e^{\delta U_{i}\left(\vx\right)}$ or $F_{i}\left(\vx\right) =1+\delta U_{i}\left(\vx\right)$, where $\delta\geqslant 0$ represents the selection intensity. These two forms of $F_{i}$\label{symbol:F_i} are equivalent under weak selection, so we do not need to explicitly choose one \citep{maciejewski:PLoSCB:2014}. In either case, weak selection is modeled as a small deviation, proportional to $\delta$, from a neutral baseline fecundity $F_{i}^{\circ}$ for all $i=1,\dots ,N$. We note that $F_{i}$ represents a general measure of the competitive abilities of $i$, which can pertain to reproduction, survival, or a combination of the two; we use the term ``fecundity'' instead of ``fitness'' to avoid confusion because the latter term is defined in many ways in the literature.

Our analysis of weak selection uses the fact that the replacement rule $p_{\left(R,\alpha\right)}\left(\vx\right)$ depends on $\vx$ because $\vx$ determines game payoffs, which then determine fecundities. The demographic parameters of individuals, which are marginal properties of $p_{\left(R,\alpha\right)}\left(\vx\right)$, also depend implicitly on fecundities. The population's fecundities are given in general by $\mathbf{F}\in\left(0,\infty\right)^{N}$, with $\mathbf{F}=\mathbf{F}^{\circ}$ when $\delta = 0$. We employ $\mathbf{F}$ as a set of intermediate variables between $\delta$ and the main object of our analysis, the mean trait frequency $\left< x_{\textrm{S}}\right>$. In our model, selection is directly mediated by one demographic parameter in particular, the replacement probability $e_{ij}\left(\vx\right)$. As in \citep{mcavoy:NHB:2020}, we define $m_{k}^{ij}\coloneqq\frac{\partial e_{ij}}{\partial F_{k}}\Big\vert_{\mathbf{F}=\mathbf{F}^{\circ}}$\label{symbol:mkij} as a measure of how individual payoffs affect competition in an evolutionary game. These terms are independent of the game itself, and they allow us to express the marginal effects of selection on the probability that $i$ replaces $j$ conveniently in terms of the payoffs of the game,
\begin{align}
e_{ij}\left(\vx\right) &= e_{ij}^{\circ} + \delta\sum_{k=1}^{N} m_{k}^{ij} U_{k}\left(\vx\right) + O\left(\delta^{2}\right) .
\end{align}
Note that we implicitly assume that $p_{\left(R,\alpha\right)}\left(\vx\right)$ is a smooth function of $\delta$ (see {\supp}, \S{SI.1.1}).

\subsection{Mean trait frequencies}
To see how our main result may be written succinctly in terms of probabilities of identity in state ($\phi_{ij}$), as well as the quantities $\pi_{i}^{\textrm{mut}}$ and $m_{k}^{ji}$, it is helpful to begin with the trait values of individuals. Fundamentally, the first-order effect of selection on the mean frequency of the social trait $\textrm{S}$ depends on expected payoffs, which in turn depend on traits. In {\supp}, we show that
\begin{align}
\left< x_{\textrm{S}} \right> &= \frac{1}{2} + \delta\frac{1}{u}\sum_{i,j,k=1}^{N} \pi_{i}^{\textrm{mut}} m_{k}^{ji} \left(\substack{\left(1-u\right)\frac{1}{2}\left<U_{k}\mid j\sim\textrm{S}\right>^{\circ} \\[1mm] + u\frac{1}{2}\left<U_{k}\right>^{\circ} \\[1mm] - \frac{1}{2}\left<U_{k}\mid i\sim\textrm{S}\right>^{\circ}}\right) + O\left(\delta^{2}\right) , \label{eq:xA_intuition2}
\end{align}
where $\left<\cdot\right>^{\circ}$ is the expectation taken over the stationary distribution of trait values across the population under neutrality. Thus, selective events involving $i$ are scaled by the mutation-weighted reproductive value, $\pi_{i}^{\textrm{mut}}$. Then, for each $j$ whose offspring can replace $i$, $m_{k}^{ji}$ quantifies how the performance of each individual $k$ influences the probability that $j$ replaces $i$. Finally, accounting for the possibility of mutation, the influx of $\textrm{S}$ to $i$ (via the reproduction of $j$) and the outflux of $\textrm{S}$ from $i$ (via the death of $i$) depend on three kinds of expected payoffs to individual $k$: $\left<U_{k}\mid j\sim\textrm{S}\right>^{\circ}$ and $\left<U_{k}\mid i\sim\textrm{S}\right>^{\circ}$ are the expected payoffs to $k$ conditioned on $j$ and $i$ having type $\textrm{S}$, respectively, while $\left<U_{k}\right>^{\circ}$ is the unconditional expected payoff to $k$.

Two kinds of marginal probabilities are required to describe the expected payoffs in \eq{xA_intuition2}. Let $p_{\left(x_{i},x_{j}\right)}^{ij}$\label{symbol:neutral_pair} be the probability that $i$ and $j$ have types $x_{i},x_{j}\in\left\{\textrm{S},\textrm{NS}\right\}$. Similarly, define $p_{\left(x_{i},x_{j},x_{k}\right)}^{ijk}$\label{symbol:neutral_triplet} for the triplet $i$, $j$, $k$. Using this notation, the expected payoff to $k$ when $j$ has type $\textrm{S}$ satisfies
\begin{align}
\frac{1}{2}\left<U_{k}\mid j\sim\textrm{S}\right>^{\circ} &= \sum_{\ell =1}^{N}\left(\substack{p_{\left(\textrm{S},\textrm{S},\textrm{S}\right)}^{jk\ell}a_{11}^{k\ell} + p_{\left(\textrm{S},\textrm{S},\textrm{NS}\right)}^{jk\ell}a_{12}^{k\ell} \\ + p_{\left(\textrm{S},\textrm{NS},\textrm{S}\right)}^{jk\ell}a_{21}^{k\ell} + p_{\left(\textrm{S},\textrm{NS},\textrm{NS}\right)}^{jk\ell}a_{22}^{k\ell}}\right) . \label{eq:Eukgivenj}
\end{align}
The prefactor of $1/2$ is the probability that $j$ has type $\textrm{S}$ under neutral drift. The left-hand side of \eq{Eukgivenj} is multiplied by $1-u$ in \eq{xA_intuition2} because if there is no mutation, then the offspring of $j$ that replaces $i$ has type $\textrm{S}$. Since we care only about the influx of $\textrm{S}$ to $i$ from $j$, we do not consider expected payoffs conditioned on $x_{j}=\textrm{NS}$. Similarly, regardless of the type of $j$, the offspring of $j$ will mutate and end up type $\textrm{S}$ with probability $u/2$. In this case, the expected payoff to $k$ is
\begin{align}
\left<U_{k}\right>^{\circ} &= \sum_{\ell =1}^{N}\left(\substack{p_{\left(\textrm{S},\textrm{S}\right)}^{k\ell}a_{11}^{k\ell} + p_{\left(\textrm{S},\textrm{NS}\right)}^{k\ell}a_{12}^{k\ell} \\ + p_{\left(\textrm{NS},\textrm{S}\right)}^{k\ell}a_{21}^{k\ell} + p_{\left(\textrm{NS},\textrm{NS}\right)}^{k\ell}a_{22}^{k\ell}}\right) , \label{eq:Euk}
\end{align}
which does not depend on the trait of $j$. Both \eq{Eukgivenj} and \eq{Euk} correspond to the influx of $\textrm{S}$ to $i$, so they have positive weights in \eq{xA_intuition2}. In contrast, the expected payoff to $k$ when $x_{i}=\textrm{S}$ satisfies
\begin{align}
\frac{1}{2}\left<U_{k}\mid i\sim\textrm{S}\right>^{\circ} &= \sum_{\ell =1}^{N}\left(\substack{p_{\left(\textrm{S},\textrm{S},\textrm{S}\right)}^{ik\ell}a_{11}^{k\ell} + p_{\left(\textrm{S},\textrm{S},\textrm{NS}\right)}^{ik\ell}a_{12}^{k\ell} \\ + p_{\left(\textrm{S},\textrm{NS},\textrm{S}\right)}^{ik\ell}a_{21}^{k\ell} + p_{\left(\textrm{S},\textrm{NS},\textrm{NS}\right)}^{ik\ell}a_{22}^{k\ell}}\right) , \label{eq:Eukgiveni}
\end{align}
which is subtracted because the replacement of $i$ by $j$ results in the loss of this $\textrm{S}$ at $i$. The prefactor of $1/2$ is the probability that $i$ has type $\textrm{S}$ under neutral drift. The left-hand side of \eq{Eukgiveni} corresponds to the outflux of $\textrm{S}$ from $i$ through death. As such, we do not need to consider terms of the form $\left<U_{k}\mid i\sim\textrm{NS}\right>^{\circ}$ because these correspond to zero outflux of $\textrm{S}$ at $i$ through death.

To present our main result in the form suggested by Eqs.~\ref{eq:xA_expansion}--\ref{eq:structure_coefficients}, we express $p_{\left(x_{i},x_{j}\right)}^{ij}$ and $p_{\left(x_{i},x_{j},x_{k}\right)}^{ijk}$ in terms of probabilities of identity by state. Simplifications are available when there are just two possible trait values. We have $p_{\left(\textrm{S},\textrm{S}\right)}^{ij}=p_{\left(\textrm{NS},\textrm{NS}\right)}^{ij}=\phi_{ij}/2$ and $p_{\left(\textrm{S},\textrm{NS}\right)}^{ij}=p_{\left(\textrm{NS},\textrm{S}\right)}^{ij}=\left(1-\phi_{ij}\right) /2$. When there are two possible trait values (and only in this case; see {\supp}), $p_{\left(x_{i},x_{j},x_{k}\right)}^{ijk}$ can be calculated using just pairwise identity probabilities as outlined in Table~\ref{tab:triplets}. These expressions, together with Eqs.~\ref{eq:Eukgivenj}--\ref{eq:Eukgiveni}, reduces \eq{xA_intuition2} to our main result, given by Eqs.~\ref{eq:xA_expansion}--\ref{eq:structure_coefficients}. An overview of the steps involved in evaluating this result in practice is given in Box~\ref{box:steps}.

\renewcommand{\arraystretch}{2}
\begin{table}
\centering
\begin{tabular}{|c|c|}
\hline
{\bf Triplet, \bm{$\left(x_{i},x_{j},x_{k}\right)$}} & {\bf Probability, \bm{$p_{\left(x_{i},x_{j},x_{k}\right)}^{ijk}$}} \\
\hhline{|=|=|}
$\left(\textrm{S},\textrm{S},\textrm{S}\right)$ & $\frac{1}{4}\left(\phi_{ij}+\phi_{ik}+\phi_{jk}-1\right)$ \\
\hline
$\left(\textrm{S},\textrm{S},\textrm{NS}\right)$ & $\frac{1}{4}\left(\phi_{ij}-\phi_{ik}-\phi_{jk}+1\right)$ \\
\hline
$\left(\textrm{S},\textrm{NS},\textrm{S}\right)$ & $\frac{1}{4}\left(-\phi_{ij}+\phi_{ik}-\phi_{jk}+1\right)$ \\
\hline
$\left(\textrm{S},\textrm{NS},\textrm{NS}\right)$ & $\frac{1}{4}\left(-\phi_{ij}-\phi_{ik}+\phi_{jk}+1\right)$ \\
\hline
$\left(\textrm{NS},\textrm{S},\textrm{S}\right)$ & $\frac{1}{4}\left(-\phi_{ij}-\phi_{ik}+\phi_{jk}+1\right)$ \\
\hline
$\left(\textrm{NS},\textrm{S},\textrm{NS}\right)$ & $\frac{1}{4}\left(-\phi_{ij}+\phi_{ik}-\phi_{jk}+1\right)$ \\
\hline
$\left(\textrm{NS},\textrm{NS},\textrm{S}\right)$ & $\frac{1}{4}\left(\phi_{ij}-\phi_{ik}-\phi_{jk}+1\right)$ \\
\hline
$\left(\textrm{NS},\textrm{NS},\textrm{NS}\right)$ & $\frac{1}{4}\left(\phi_{ij}+\phi_{ik}+\phi_{jk}-1\right)$ \\
\hline
\end{tabular}
\caption{Calculating the probability that $i$, $j$, and $k$ have types $x_{i}$, $x_{j}$, and $x_{k}$, respectively, under neutral drift when there are just two possible trait values, $\textrm{S}$ and $\textrm{NS}$. For each triplet of traits, $\left(x_{i},x_{j},x_{k}\right)$, the probability $p_{\left(x_{i},x_{j},x_{k}\right)}^{ijk}$ is a function of the identity-by-state probabilities $\phi_{ij}$, $\phi_{ik}$, and $\phi_{jk}$.\label{tab:triplets}}
\end{table}
\renewcommand{\arraystretch}{1}

\renewcommand{\figurename}{Box.}
\begin{figure}
\centering
\begin{tcolorbox}
\centering
\begin{enumerate}

\item Specify $p_{\left(R,\alpha\right)}\left(\vx\right)$, the probability of choosing replacement event $\left(R,\alpha\right)$ in state $\vx$;

\item Compute the marginal probability that $i$ transmits its offspring to $j$, $e_{ij}\left(\vx\right)$;

\item Compute the death probability of $i$, $d_{i}\left(\vx\right) =\sum_{j=1}^{N}e_{ji}\left(\vx\right)$;

\item Compute the matrix $A$ for the (neutral) ancestral random walk, with $A_{ij}=e_{ji}^{\circ}/d_{i}^{\circ}$;

\item Compute $z_{i}^{\textrm{mut}}=\frac{1}{N}\sum_{j=1}^{N} u \left(\left(I-\left(1-u\right) A\right)^{-1}\right)_{ji}$, where $N$ is the population size and $u$ is the mutation probability, then compute $\pi_{i}^{\textrm{mut}}=z_{i}^{\textrm{mut}}/d_{i}^{\circ}$;

\item Compute the probability that $i$ and $j$ are identical by state, $\phi_{ij}$, using \eq{phi_ij_recurrence};

\item Compute the marginal effect of $k$ on $i$ replacing $j$ at neutrality, $m_{k}^{ij}=\frac{\partial e_{ij}}{\partial F_{k}}\Big\vert_{\mathbf{F}=\mathbf{F}^{\circ}}$, where $F_{k}$ is the fecundity of $k$ and $\mathbf{F}$ is the fecundity vector of the entire population;

\item Compute the structure coefficients $K_{1}^{k\ell}$ and $K_{2}^{k\ell}$ using steps (5)--(7) and \eq{structure_coefficients};

\item Use step (8), together with the payoffs, to compute $\frac{d}{d\delta}\Big\vert_{\delta =0}\left< x_{\textrm{S}}\right>$ by means of \eq{xA_expansion}.

\end{enumerate}
\end{tcolorbox}
\caption{Instructions for calculating the effects of weak selection on the mean frequency of $\textrm{S}$.\label{box:steps}}
\end{figure}
\renewcommand{\figurename}{Figure.}
\setcounter{figure}{0}

\eq{xA_expansion} has a form similar to the structure-coefficient theorem of \citet{tarnita:JTB:2009}. In fact, our result includes this structure-coefficient theorem as a special case. Note that $a_{xy}^{ij}$ implicitly encodes which individuals interact and how payoffs are aggregated through the dependence on $i$ and $j$. For example, we would define the payoffs to be zero for non-interacting pairs of individuals. In the case of a symmetric game described by \eq{symmetric_payoff_matrix}, an interaction structure for the population needs to be specified. Let $\left(\Omega_{ij}\right)_{i,j=1}^{N}$\label{symbol:omega} be a matrix representing this interaction structure. Then, if the individuals at locations $i$ and $j$ have types $\textrm{S}$ and $\textrm{NS}$, respectively, $i$ gets payoff $\Omega_{ij}a_{12}$ and $j$ gets payoff $\Omega_{ji}a_{21}$ from this interaction. If $i$ and $j$ do not interact, then $\Omega_{ij}=0$. For every individual, these scaled payoffs are summed over all interaction partners.

A graph-structured population can be described using an adjacency matrix $\left(w_{ij}\right)_{i,j=1}^{N}$\label{symbol:w_ij}, with $w_{ij}=1$ if $i$ and $j$ are neighbors and $w_{ij}=0$ otherwise. In our model, this would then constrain the replacement rule, with (for example) $e_{ij}\left(\vx\right)=0$ whenever $w_{ij}=0$. A common approach is to set the interaction matrix $\Omega_{ij}=w_{ij}$ to indicate that individuals interact only with neighbors and accumulate the resulting payoffs, e.g. additively. Another approach is to let $\Omega_{ij}=w_{ij}/w_{i}$, where $w_{i}$ is the total number of neighbors of $i$, yielding average payoffs. Both kinds of payoffs are commonly used in evolutionary game theory \citep[see][]{maciejewski:PLoSCB:2014}. Whatever the interaction structure is for a particular population, the structure-coefficient theorem for symmetric games may be obtained from the result above by setting $a_{11}^{ij}=\Omega_{ij}a_{11}$, $a_{12}^{ij}=\Omega_{ij}a_{12}$, $a_{21}^{ij}=\Omega_{ij}a_{21}$, and $a_{22}^{ij}=\Omega_{ij}a_{22}$ in \eq{xA_expansion}.

Our results can be extended to games with $n>2$ strategies. For general matrix games with two strategies, as considered above, or for additive matrix games with $n\geqslant 3$ strategies, the complexity of evaluating the structure-coefficient theorem is bounded by solving a linear system of size $O\left(N^{2}\right)$. For non-additive matrix games with $n\geqslant 3$ strategies, this linear system is of size $O\left(N^{3}\right)$. We give a complete description of these extensions to $n$ strategies in {\supp}.

\section{Applications to producers of social goods}
We conclude with an application of our theoretical results to the evolution of producers of social goods. We are interested in how the evolutionary advantages/disadvantages of making ff-goods versus pp-goods are affected by mutation, for example whether the noise introduced by mutation always favors one kind of good over another. We investigate how these advantages/disadvantages change as mutation goes from weak to strong, and we use simulations to verify the accuracy of our numerical predictions for the mean frequency of the social trait at equilibrium.

We assume that, at each point in time, neighbors on a graph interact and receive a total payoff according to the nature of the goods generated by producers. Let $\left(w_{ij}\right)_{i,j=1}^{N}$ be the adjacency matrix of the population structure, which here is identical to the interaction structure for the game. A producer at location $i$ pays $C_{ij}$ in order to provide $j$ with a benefit of $B_{ij}$. For pp-goods, $B_{ij}=bw_{ij}$ and $C_{ij}=cw_{ij}$; for ff-goods, $B_{ij}=bw_{ij}/w_{i}$ and $C_{ij}=cw_{ij}/w_{i}$, where $w_{i}$ is the number of neighbors of $i$. In other words, we have $a_{11}^{ij}=B_{ji}-C_{ij}$, $a_{12}^{ij}=-C_{ij}$, $a_{21}^{ij}=B_{ji}$, and $a_{22}^{ij}=0$.

An individual's total payoff, $U_{i}\left(\vx\right)$, obtained by summing over all interactions, is converted to fecundity, $F_{i}\left(\vx\right) =e^{\delta U_{i}\left(\vx\right)}$, where $\delta\geqslant 0$ is the intensity of selection and $\vx$ is the current state of the population. The population then updates according to the death-Birth (dB) rule, in which an individual is first selected uniformly at random for death and then the neighbors compete, based on their relative fecundities, to fill the vacancy. The neighbors in this reproduction step are also determined by $\left(w_{ij}\right)_{i,j=1}^{N}$. In this case, the probability that $i$ replaces $j$ when the state is $\vx$ is
\begin{align}
e_{ij}\left(\vx\right) &= \frac{1}{N} \frac{F_{i}\left(\vx\right) w_{ij}}{\sum_{k=1}^{N}F_{k}\left(\vx\right) w_{kj}} . \label{eq:eij_db}
\end{align}

Under neutral drift, the probability that $i$ replaces $j$ is $e_{ij}^{\circ}=\left(w_{ji}/w_{j}\right) /N$, and the probability that $i$ dies in one step of the process is $d_{i}^{\circ}=1/N$. The probability of moving from $i$ to $j$ in one step of the ancestral random walk is then $A_{ij}=w_{ij}/w_{i}$. The stationary distribution for this random walk satisfies $z_{i}=w_{i}/\sum_{j=1}^{N}w_{j}$, which gives $i$ a reproductive value of $\pi_{i}=Nw_{i}/\sum_{j=1}^{N}w_{j}$. For $u>0$, the relevant notion of reproductive value, $\pi_{i}^{\textrm{mut}}$, has a more complicated expression but can easily be found by calculating $z_{i}^{\textrm{mut}}$ from $\left(A_{ij}\right)_{i,j=1}^{N}$ (using \eq{mui}) and letting $\pi_{i}^{\textrm{mut}}=z_{i}^{\textrm{mut}}/d_{i}^{\circ}$.

The pairwise identity-by-state probabilities, $\phi$, satisfy $\phi_{ii}=1$ for every $i=1,\dots ,N$. For $i\neq j$, since exactly one individual is replaced in each time step under dB updating, \eq{phi_ij_recurrence} gives
\begin{align}
\phi_{ij} &= u\frac{1}{2} + \left(1-u\right)\frac{1}{2}\sum_{k=1}^{N} A_{ik} \phi_{kj} + \left(1-u\right)\frac{1}{2}\sum_{k=1}^{N} A_{jk} \phi_{ik} . \label{eq:db_recurrence}
\end{align}
In other words, to determine the probability that $i$ is identical by state to $j\neq i$, we wait until either $i$ or $j$ is replaced by the offspring of some parent $k$. The first term on the right-hand side of \eq{db_recurrence} represents a mutation from the parent, which happens with probability $u$ and leads to $i$ and $j$ being identical by state with probability $1/2$. Otherwise, with probability $1-u$ there is no mutation, $i$ is the one replaced with probability $1/2$, $k$ is the parent with probability $A_{ik}$, and $k$ is identical by state to $j$ with probability $\phi_{kj}$. The same logic applies to the last term, when $j$ is the one replaced.

Finally, the marginal effects of selection on replacement satisfy $m_{i}^{ij}=A_{ji}\left(1-A_{ji}\right) /N$ and $m_{k}^{ij}=-A_{ji}A_{jk}/N$ if $k\neq i$. These marginal effects can be calculated directly from \eq{eij_db}.

Putting all of these ingredients together, we see that \eq{xA_intuition2} (or \eq{xA_expansion}) takes the form
\begin{align}
\left< x_{\textrm{S}} \right> &= \frac{1}{2} + \delta\frac{1-u}{2Nu}\sum_{i,j,k=1}^{N} \pi_{i}^{\textrm{mut}} A_{ij} A_{ik} \Bigg(\underbrace{\sum_{\ell =1}^{N}\left(-\phi_{jj}C_{j\ell}+\phi_{j\ell}B_{\ell j}\right)}_{\textrm{expected payoff to $j$, given $j$ is S}} \nonumber \\
&\qquad\qquad\qquad\qquad\qquad\qquad\qquad -\underbrace{\sum_{\ell =1}^{N}\left(-\phi_{jk}C_{k\ell}+\phi_{j\ell}B_{\ell k}\right)}_{\textrm{expected payoff to $k$, given $j$ is S}}\Bigg) +O\left(\delta^{2}\right) . \label{eq:db_main_condition}
\end{align}
This expression has a simple interpretation: As with \eq{xA_intuition2}, the consequences of selection at location $i$ are scaled by the mutation-weighted reproductive value, $\pi_{i}^{\textrm{mut}}$. Each pair of other locations $j$ and $k$ affects what happens at $i$ according first to the probabilities $A_{ij}$ and $A_{ik}$. Next, as depicted in \fig{db_condition}, the selective effect of a producer at $j$ depends on the expected payoff to $j$ compared to that of the other neighbor at $k$. Here, $i$ represents the individual chosen for death, and producers are favored when a neighboring producer is more successful than its other competitors, on average.

\begin{figure}
\centering
\includegraphics[width=0.8\textwidth]{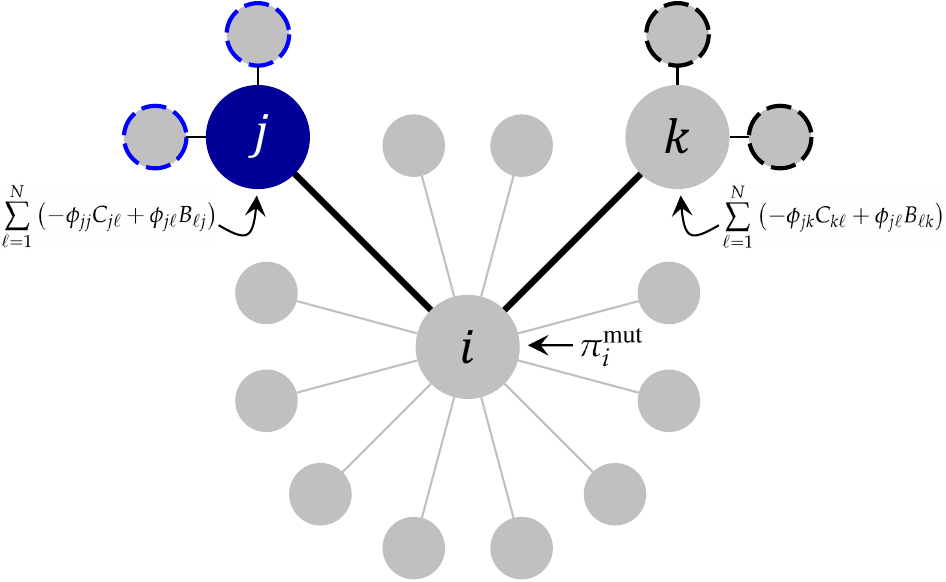}
\caption{Intuition for the effects of selection on trait frequency. A location, $i$, is first chosen with probability proportional to its mutation-weighted reproductive value, $\pi_{i}^{\textrm{mut}}$. Although all individuals are replaced with the same probability under death-Birth updating, $\pi_{i}^{\textrm{mut}}$ quantifies how important $i$ is as a location for a social trait to be propagated. If $i$ is replaced, then neighbors compete to fill the vacancy. Since we are concerned with producers of a social good (blue), the relevant comparison is between a neighboring producer and other neighbors. We choose two neighbors, $j$ and $k$, with probability $A_{ij}$ and $A_{ik}$, respectively, which represent the probabilities of $j$ and $k$ replacing $i$ under neutral drift (conditioned on $i$ being replaced). If a producer is placed at $j$, then the probability that another node, $\ell$, is a producer is equal to the probability that $j$ and $\ell$ are identical by state, $\phi_{j\ell}$. Using this distribution, we can calculate the expected payoffs for the two neighbors, $j$ and $k$, which are depicted next to each of the corresponding nodes. It is the difference between the payoff to the producer, $j$, and the other neighbor, $k$, that determines the effects of selection on the mean frequency of producers. The larger this difference is, the more strongly favored producers are.\label{fig:db_condition}}
\end{figure}

As \fig{ff_vs_pp}A shows, producers of ff-goods can be favored even when $b<c$. This seemingly paradoxical behavior is consistent to what is observed for ff-goods from an analysis of fixation probabilities in the weak-mutation limit \citep{mcavoy:NHB:2020}, namely that they benefit from the existence of a small number of highly-connected hubs. When producers are rare and are introduced at a well-connected location, they are selected against, but drift can lead to the establishment of a small colony of reciprocating producers among the neighboring nodes. Even when $b<c$, this can result in a successful producer at the hub because they receive multiple benefits (or partial benefits due to division) and pay only a single cost. Having a relatively large payoff, selection then favors the propagation of the hub's offspring (producers) to the periphery, provided the mutation probability is not too large, which serves to strengthen the hub's selective advantage. Similarly, when producers are widespread, a non-producer mutant is more likely to arise as a neighbor of this highly-connected individual than at the hub itself. Since the hub then has an extremely large payoff owing to the existence of many separate donations, each of benefit $b$ (or close to $b$), the hub can easily replace the mutant non-producer by a producer. Producers therefore thrive by taking over highly-connected hubs and populating the neighborhoods with copies of themselves. These neighboring locations are likely to hold producers for long periods of time, which gives them more opportunities to spread their offspring to other locations in the population, even when at a reproductive disadvantage relative to other competitors.

\begin{figure}
\centering
\includegraphics[width=0.8\textwidth]{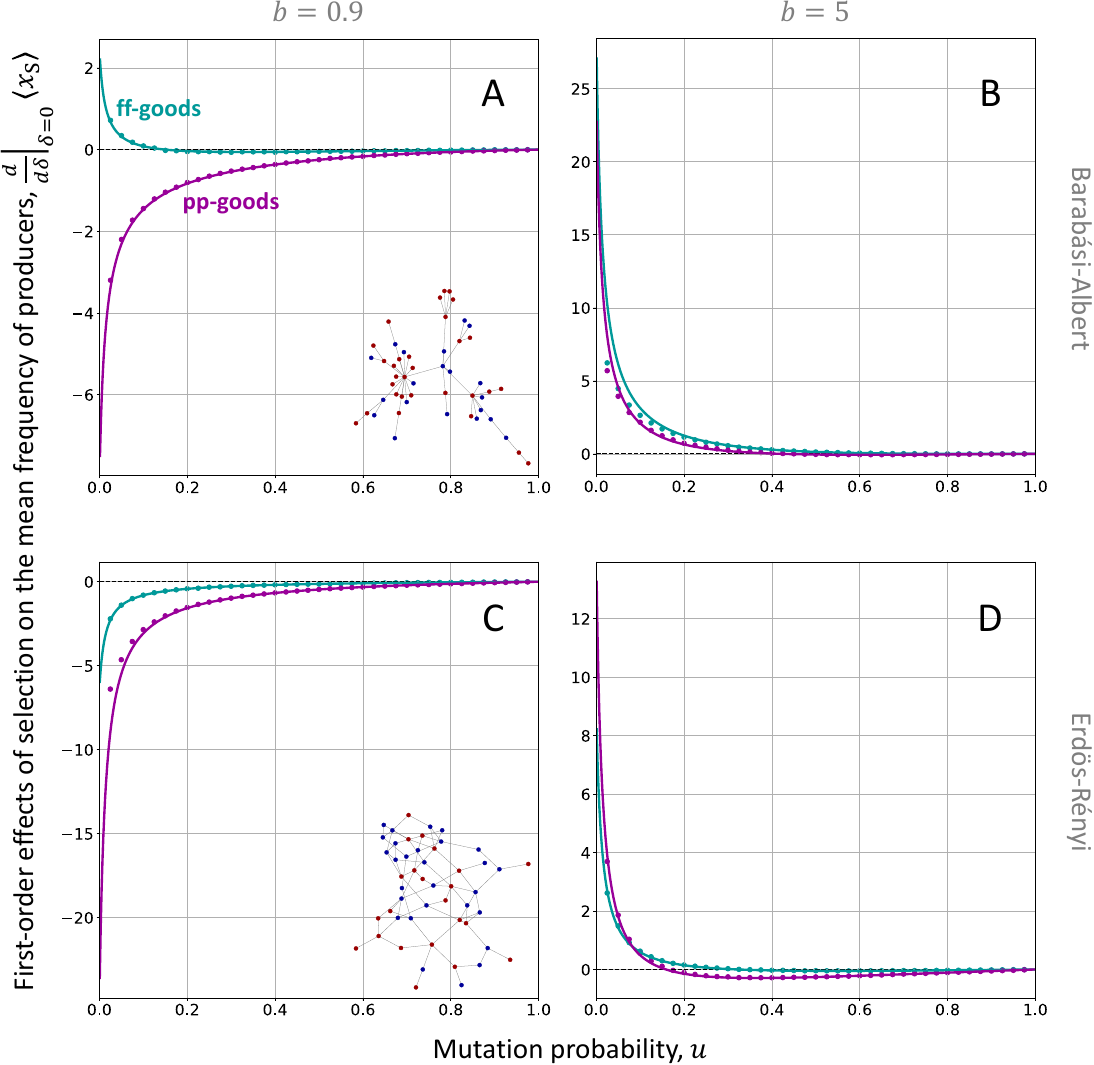}
\caption{Effects of selection on the mean frequencies of simple prosocial behaviors. The solid lines indicate calculations performed using the structure-coefficient theorem (Eqs.~\ref{eq:xA_expansion}--\ref{eq:structure_coefficients}). Dots depict the simulated first-order effects of selection, $\left(\left< x_{\textrm{S}}\right> -1/2\right) /\delta$, where $\left<x_{\textrm{S}}\right>$ is calculated for a selection intensity of $\delta =0.05$ by observing the mean frequency of $\textrm{S}$ over $10^{8}$ generations. A and B show a single Barab\'{a}si-Albert (preferential-attachment) graph of size $N=50$, with $b<c$ in A and $b>c$ in B. Even when $b<c$ (here, $b=0.9$ and $c=1$), we see that producers of ff-goods can be favored in the mutation-selection equilibrium. C and D show a single Erd\"{o}s-R\'{e}nyi graph of size $N=50$, with $b<c$ in C and $b>c$ in D. When $b>c$ (here, $b=5$ and $c=1$), whether selection favors producers of ff-goods more than those of pp-goods can depend on the mutation probability. In this example, producers of pp-goods are favored more when the mutation probability is small, but this ranking reverses for slightly larger values of $u$ (even when both are favored relative to neutral drift). The scales of the vertical axes vary to best illustrate the results.\label{fig:ff_vs_pp}}
\end{figure}

However, this advantage enjoyed by producers of ff-goods diminishes to zero and even becomes (slightly) negative as the mutation probability increases. For one thing, a high mutation probability means that a central producer is more likely to propagate non-producers to neighboring positions. It also means that when a highly-connected producer dies, it is more likely to be replaced by a non-producer via mutation alone, even when all (or most) neighbors are producers. The detrimental effects of increased mutation probabilities are seen in all examples for which producers of ff-goods or pp-goods are favored when mutation probabilities are small (\fig{ff_vs_pp}A,B,D), both when $b<c$ and $b>c$.

When the degree distribution is more tightly concentrated around the mean, the population is less likely to have highly-connected hubs surrounded by individuals of small degree. As a result, neighboring producers must share their benefit among a larger number of two-step neighbors. Consequently, the argument we made for why producers can be favored even when $b<c$ on Barab\'{a}si-Albert graphs breaks down on Erd\"{o}s-R\'{e}nyi graphs (see \fig{ff_vs_pp}C). In addition, \fig{ff_vs_pp}D shows that even when $b>c$ and producers of both kinds of goods are favored (i.e. for mutation probabilities that are not too high), the ranking of which one is favored more can be reversed by changing the mutation probability. Thus, the noise introduced by mutation does not affect pp-goods and ff-goods uniformly. \fig{ff_vs_pp}A,B,D also demonstrates that for a single kind of good, there need not be any monotonicity in the effects of selection as the mutation probability increases: increasing $u$ can be harmful to the abundance of producers when $u$ is small but beneficial to the abundance of producers when $u$ is larger. Finally, and not surprisingly, selection via any kind of game becomes completely ineffective as $u \to 1$.

The differences between these two kinds of population structures, in terms of how they affect mutation-selection dynamics, can be summarized by their associated probabilities of identity by descent (Eqs.~\ref{eq:db_recurrence}--\ref{eq:db_main_condition} and \fig{db_condition}). In \fig{db_condition}, if $j$, a producer, has neighbors (blue dashed circles) who are themselves not highly connected (e.g. for the structure in \fig{ff_vs_pp}A,B), then $j$ is likely to replace these neighbors when they are selected for death. As a result, $j$ and its immediate neighbors are more likely to be identical by state, provided $u$ is not too large, which means that the neighbors produce benefits and provide relatively large fractions of their benefits to $j$. In contrast, the neighbors (black dashed circles) of a competitor, $k$, are generally less likely to be identical by state to $j$ because paths from $j$ to $k$ (and $k$'s neighbors) go through much more highly connected hubs. Thus, even though $k$ itself is less likely to pay a cost, $c$, than $j$, $k$ also receives fewer benefits than $j$, on average, resulting in a competitive advantage for $j$ in replacing $i$. For structures like that of \fig{ff_vs_pp}C,D, the same basic logic of \fig{db_condition} holds, but now the degree distribution is different, and we see fewer instances of highly connected nodes being surrounded by much less connected individuals. As a result, the ratio $b/c$ must generally be larger in order to offset lower probabilities of identity by state, larger net costs paid by $j$ (in the case of pp-goods), and smaller shares of goods being donated to $j$ by neighbors (in the case of ff-goods).

In {\supp}, we include a version of \fig{ff_vs_pp} with even weaker selection ($\delta =0.01$ instead of $\delta = 0.05$), which, as expected, shows even better agreement with the calculations (see \fig{ff_vs_pp_weaker}). We also illustrate the effects of larger values of $\delta$ (Figs.~\ref{fig:stronger_selection}--\ref{fig:stronger_selection_same_plot}) to assess the departures from these predictions as the strength of selection increases.

\section{Discussion}
Since its genesis in modeling animal conflict \citep{maynardsmith:Nature:1973,maynardsmith:CUP:1982}, evolutionary game theory has proven itself to be a tremendous resource in the study of social evolution. But any useful modeling framework necessitates the development of methods and tools for its analysis. Deterministic approaches, such as those based on the replicator equation \citep{taylor:MB:1978}, render the resulting models amenable to standard techniques from the theory of dynamical systems \citep{hofbauer:CUP:1998} but may neglect important evolutionary phenomena. Stochastic perspectives on modeling population dynamics can incorporate the effects of drift and fine-grained spatial structure, but several practical complications also arise. Outside of the realm of weak selection, there need not be any efficient algorithm for evaluating a population's evolutionary dynamics, in general \citep{ibsenjensen:PNAS:2015}, a conclusion that has prompted other kinds of computational techniques, such as the development of more efficient algorithms for numerical simulation \citep{hindersin:SR:2019}. Our contribution here is a method to study frequencies of types at stationarity under weak selection.

Our focus on weak selection involves the assumption that replacement probabilities are smooth functions of the selection intensity, $\delta$. One consequence is that the system under small $\delta$ is qualitatively similar to the system under neutral drift ($\delta =0$). For the latter, we assume that replacement probabilities are independent of the state, $\vx$, in the absence of selection. We also make a technical assumption that the population evolves as a coherent whole, which ensures that there exists a unique stationary distribution to study in the first place. These assumptions on the model are formally stated in {\supp}. Of course, although they are biologically reasonable, these assumptions also limit the scope of the model and should be taken into consideration when using our results to go beyond the examples illustrated here.

Another limitation is that, although the population can have arbitrary size and structure, both are fixed as the evolutionary process unfolds. Dynamic population structures are both realistic and recognized (largely through numerical simulations) as having a strong impact on the evolution of prosocial traits like cooperation \citep{wardil:SR:2014,khoo:SR:2016,akcay:NC:2018,khoo:JCN:2018}. Related to this point, an interesting extension would account for feedback between individuals and the environment as a way to capture ecosystem engineering \citep{hauert:JTB:2019,estrela:TEE:2019,tilman:NC:2020}. From a computational perspective, it would also be intriguing to explore how these components affect the complexity of quantitative analysis. These considerations are among many that would be necessary for a deeper understanding of social evolution \citep{akcay:AN:2020}.

In our examples, we find that the effects of selection on producers of two different kinds of social goods can be quite sensitive to mutation probabilities. This effect is conceptually similar to the ways in which the infidelity of Mendelian transmission of diploid phenotypes can alter evolutionary game theoretic predictions relative to those based on simple haploid systems \citep{grafen:BP:2014,rubin:BP:2016}. We compare goods that are produced with benefits and costs in proportion to the number of interaction partners (pp-goods) to goods that are produced with fixed benefits and costs then get distributed among interaction partners (ff-goods). We show that increasing the mutation probability can reverse the direction of selection and can change the relative ranking of the two types, even when producers of both kinds of goods are favored by selection. Qualitatively, one thing that the analysis of social goods on heterogeneous graphs yields is the fact that mean frequencies need not be monotonic in the mutation probability. \citet{traulsen:PNAS:2009} also found non-monotonicity but only in games with more than two traits. \citet{debarre:JTB:2017} observed only monotonic effects of mutation probabilities on mean frequencies, but the structures considered there are homogeneous, which can obscure the difference between social goods. Non-monotonicity may be inferred from the results of \citet{debarre:DGA:2020} for small to moderate migration rates (changing who interacts with whom) in an island model of population structure.

The importance of high mutation rate has long been recognized in evolutionary theory. If it is high enough, the resulting ``quasispecies'' crosses an error threshold beyond which adaptation is impossible \citep{eigen:TJPC:1988}. This model has been used to study human immunodeficiency viruses, which are highly diverse owing to rapid mutation \citep{nowak:TEE:1992}. The replicator equation has been generalized to include mutation \citep{page:JTB:2002}, bridging the gap between replicator dynamics \citep{taylor:MB:1978} for frequency-dependent selection (without mutation) and quasispecies evolution under strong mutation and frequency-\emph{independent} selection. A notable application is to language learning in populations \citep{komarova:JTB:2004}. In simple (genetically single-locus) models, increased mutation may just reduce the efficacy of selection, eventually resulting in neutral evolution, as we find here. But it is important to recognize that empirical situations are likely more complicated because increased mutation rates apply to all loci across the genome and many mutations are deleterious; for one example, see \citet{sprouffske:PLOSGEN:2018}.

When the population is heterogeneous, the distribution of social goods naturally gives rise to asymmetric games. Asymmetric contests have an extensive history in theoretical biology, dating back nearly half a century. In a pioneering theoretical work on animal behavior, \citet{maynardsmith:AB:1976} showed that while evolutionarily stable states are frequently mixed in symmetric contests, variability in traits and behaviors is actually the exception in asymmetric contests. Furthermore, they argued that asymmetric interactions should be taken as the norm, either as the result of organismal differences like size or strength, or as the result of ``uncorrelated'' asymmetries such as those that arise when one individual discovers a resource and another is a latecomer.

One takeaway from earlier studies on asymmetric games, as well as from our own examples of social goods, is that there is unequivocally both a loss of generality and a loss of realism in approximating interactions by symmetric games. The assumption of homogeneous interactions severely limits the scope of a model. This aspect of social evolution is still relatively underexplored in the literature. For example, the spatially-mediated heterogeneity arising in the production of ff-goods underscores, among other things, the need for an asymmetric version of the structure-coefficient theorem. Even though no mathematical model can accurately capture every aspect of a population evolving in the natural world, striving for refinements that are simultaneously realistic and amenable to analysis is necessary for a deeper understanding of population dynamics.

\section*{Acknowledgments}
The authors thank Florence D\'{e}barre for a wealth of insightful comments on this work. A. M. was supported by a Simons Postdoctoral Fellowship (Math+X) at the University of Pennsylvania.

\section*{Code availability}
Code supporting the findings of this study may be found at https://github.com/alexmcavoy/sigma.

\setcounter{equation}{0}
\setcounter{figure}{0}
\setcounter{section}{0}
\setcounter{table}{0}
\renewcommand{\thesection}{SI.\arabic{section}}
\renewcommand{\thesubsection}{SI.\arabic{section}.\arabic{subsection}}
\renewcommand{\theequation}{SI.\arabic{equation}}
\renewcommand{\thetable}{SI.\arabic{table}}
\renewcommand{\thefigure}{SI.\arabic{figure}}

\section{Modeling mutation-selection processes}

\subsection{Replacement and mutation}\label{sec:replacement_rules}
We consider the evolutionary dynamics of $n$ strategies, labeled $1,\dots ,n$, in a population of finite, fixed size, $N$. The state of the population is specified by a configuration of strategies, $\vx\in\left\{1,\dots ,n\right\}^{N}$. We use the notion of a \emph{replacement event} \citep{allen:JMB:2014,allen:JMB:2019} to describe a general model of reproduction. A replacement event is a pair $\left(R,\alpha\right)$, where $R\subseteq\left\{1,\dots ,N\right\}$ is the set of individuals who will die and be replaced by offspring in a given time step, where $\alpha :R\rightarrow\left\{1,\dots ,N\right\}$ is a parentage (offspring-to-parent) map. In state $\vx$, a replacement event is chosen based on a distribution, $\left\{p_{\left(R,\alpha\right)}\left(\vx\right)\right\}_{\left(R,\alpha\right)}$, known as a \emph{replacement rule}. That is, the probability that replacement event $\left(R,\alpha\right)$ occurs is equal to  $p_{\left(R,\alpha\right)}\left(\vx\right)$ when the population's state is $\vx$. We use $i$, $j$, $k$, etc., to denote individuals at particular sites in this general model of population structure.

If $i\in R$, then the strategy of $i$ in the next step depends on the strategy of its parent, which is denoted $\alpha\left(i\right)$. For a fixed mutation probability, $u>0$, $i$ inherits the strategy of $\alpha\left(i\right)$ with probability $1-u$. With probability $u$, a mutation occurs and $i$ takes on a new strategy in the set $\left\{1,\dots ,n\right\}$ uniformly at random. If $i\not\in R$, in other words if $i$ lives through this time step, its strategy remains unchanged. Thus, given a particular replacement event $\left(R,\alpha\right)$ and the current state of the population $\left(\vx\right)$, the probability that $i$ has strategy $y$ in the next time step is
\begin{align}
\mathbb{P}\left[ X_{i}^{t+1} = y\mid \left(R,\alpha\right),\ \vX^{t}=\vx \right] &= 
\begin{cases}
\mathbb{1}_{x_{i}^{t}}\left(y\right) & i\not\in R , \\
& \\
\left(1-u\right)\mathbb{1}_{x_{\alpha\left(i\right)}^{t}}\left(y\right) + u\frac{1}{n} & i\in R ,
\end{cases}
\end{align}
where $\mathbb{1}_{x}\left(y\right)$ is $1$ if $x=y$ and $0$ otherwise.

Our notation corresponds to that of \citep{allen:JMB:2014,allen:JMB:2019}, with some minor differences; see \tab{notationtable} in the main text and Table 1 of \citet{allen:JMB:2019}. We note here that $p_{\left(R,\alpha\right)}\left(\vx\right)$ depends on the configuration of strategies, $\vx$, by way of a game, which determines the fecundities of individuals. A pairwise interaction matrix $\left(\Omega_{ij}\right)_{i,j=1}^{N}$ specifies who plays with whom, and the resulting payoffs to  individuals determine their fecundities, which we denote by a vector $\mathbf{F}$. A global parameter $\delta\geqslant 0$, representing the strength or intensity of selection, determines the extent to which an individual's payoff in the game affects its fecundity. When $\delta = 0$, the population evolves neutrally. In what follows, we use the dependence on fecundities to quantify the first-order ($\delta$ near zero) effects of selection on key components of population heterogeneity.

Although the replacement rule $p_{\left(R,\alpha\right)}\left(\vx\right)$ can be used to model arbitrary population structures, our analysis depends on the following three mild assumptions. Note that the parentage map $\alpha :R\rightarrow\left\{1,\dots ,N\right\}$ may be augmented to include the cases $i\not\in R$, by the definition
\begin{align}
\widetilde{\alpha}\left(i\right) &= 
\begin{cases}
i & i\not\in R , \\
& \\
\alpha\left(i\right) & i\in R 
\end{cases}
\end{align}
where now $\widetilde{\alpha}:\left\{1,\dots ,N\right\}\rightarrow\left\{1,\dots ,N\right\}$.
We assume that:
\begin{enumerate}

\item For every $\vx$ and $\left(R,\alpha\right)$, $p_{\left(R,\alpha\right)}\left(\vx\right)$ is a smooth function of $\delta$;

\item Under neutral drift ($\delta =0$), $p_{\left(R,\alpha\right)}\left(\vx\right)$ is independent of the state, $\vx$;

\item There exists $i\in\left\{1,\dots ,N\right\}$, an integer $m\geqslant 1$, and a set of replacement events $\left\{\left(R_{k},\alpha_{k}\right)\right\}_{k=1}^{m}$ such that \emph{(i)} $p_{\left(R_{k},\alpha_{k}\right)}\left(\vx\right) >0$ for every $k\in\left\{1,\dots ,m\right\}$ and $\vx\in\left\{1,\dots ,n\right\}^{N}$; \emph{(ii)} $i\in R_{k}$ for some $k\in\left\{1,\dots ,m\right\}$; and \emph{(iii)} $\widetilde{\alpha}_{1}\circ\widetilde{\alpha}_{2}\circ\cdots\circ\widetilde{\alpha}_{m}\left(j\right) =i$ for every $j\in\left\{1,\dots ,N\right\}$.

\end{enumerate}
Thus, we assume that $p_{\left(R,\alpha\right)}\left(\vx\right)$ deviates from the corresponding neutral replacement rule by order $\delta$ when $\delta$ is small. We use $p_{\left(R,\alpha\right)}^{\circ}$ to denote the probability of $\left(R,\alpha\right)$ under neutrality ($\delta =0$). The third assumption has been called the ``fixation axiom'' \citep{allen:JMB:2019} and guarantees that at least one individual can generate a lineage that takes over the population. It also guarantees that the resulting Markov chain on $\left\{1,\dots ,n\right\}^{N}$ has a unique mutation-selection stationary (MSS) distribution, $\pi_{\MSS}$.

Our results concerning the effects of selection depend on simple demographic parameters of individuals that are marginal properties of a given replacement rule. In state $\vx$, the probability that $i$ produces an offspring that replaces $j$ in the next step is
\begin{align}
e_{ij}\left(\vx\right) &\coloneqq \sum_{\substack{\left(R,\alpha\right) \\ j\in R,\ \alpha\left(j\right) =i}} p_{\left(R,\alpha\right)}\left(\vx\right) .
\end{align}
Note that this means that $j$ died in that step. Under neutral drift, we denote this quantity by $e_{ij}^{\circ}$, which does not depend on $\vx$.

\subsection{Quantifying the mean change in a strategy's abundance}\label{sec:mean_change}
Given that the state is $\vx$ at time $t$, the probability that $i$ has strategy $y$ at time $t+1$ is
\begin{align}
\mathbb{P} &\left[ X_{i}^{t+1}=y \mid \vX^{t}=\vx \right] \nonumber \\
&= \left(1-d_{i}\left(\mathbf{x}\right)\right)\mathbb{1}_{y}\left(x_{i}\right) + \sum_{j=1}^{N}e_{ji}\left(\mathbf{x}\right)\left(\left(1-u\right)\mathbb{1}_{y}\left(x_{j}\right) + u\frac{1}{n}\right) .
\end{align}
It follows that the expected value of the $v$-weighted abundance of $y$ in the step after state $\vx$ is
\begin{align}
\mathbb{E} &\left[ \sum_{i=1}^{N}v_{i}\mathbb{1}_{y}\left(X_{i}^{t+1}\right) \mid \vX^{t}=\vx \right] \nonumber \\
&= \sum_{i=1}^{N}v_{i}\mathbb{P}\left[ X_{i}^{t+1}=y \mid \vX^{t}=\vx \right] \nonumber \\
&= \sum_{i=1}^{N}v_{i}\mathbb{1}_{y}\left(x_{i}\right) + \sum_{i,j=1}^{N}v_{i}e_{ji}\left(\mathbf{x}\right)\left(\left(1-u\right)\mathbb{1}_{y}\left(x_{j}\right) + u\frac{1}{n}-\mathbb{1}_{y}\left(x_{i}\right)\right) .
\end{align}
Given that the state is $\vx$, the expected change in the $v$-weighted abundance of $y$ in one step is then
\begin{align}
\Delta_{v,y}\left(\vx\right) &\coloneqq \sum_{i,j=1}^{N}v_{i}e_{ji}\left(\mathbf{x}\right)\left(\left(1-u\right)\mathbb{1}_{y}\left(x_{j}\right) + u\frac{1}{n}-\mathbb{1}_{y}\left(x_{i}\right)\right) . \label{eq:delhatdef}
\end{align}
Taking the expectation of both sides of \eq{delhatdef} with respect to the MSS distribution, we find that $\mathbb{E}_{\MSS}\left[\Delta_{v,y}\left(\vX\right)\right] =0$ (since the expected change in a quantity must be zero in a stationary distribution). Differentiating both sides of this equation with respect to $\delta$ at $\delta =0$ gives
\begin{align}
-\frac{d}{d\delta}\Bigg\vert_{\delta =0}\mathbb{E}_{\MSS}\left[\Delta_{v,y}^{\circ}\left(\vX\right)\right] &= \mathbb{E}_{\MSS}^{\circ}\left[\frac{d}{d\delta}\Bigg\vert_{\delta =0}\Delta_{v,y}\left(\vX\right)\right] . \label{eq:delprime}
\end{align}
We consider the left- and right-hand sides of this equation separately.

\subsubsection{Left-hand side of \eq{delprime}, $-\frac{d}{d\delta}\Bigg\vert_{\delta =0}\mathbb{E}_{\MSS}\left[\Delta_{v,y}^{\circ}\left(\vX\right)\right]$}
We first note that under neutral drift ($\delta =0$),
\begin{align}
\Delta_{v,y}^{\circ}\left(\vx\right) &= \sum_{i,j=1}^{N}\left(\left(1-u\right) v_{j} e_{ij}^{\circ} -v_{i} e_{ji}^{\circ}\right)\mathbb{1}_{y}\left(x_{i}\right) + u\frac{1}{n}\sum_{i,j=1}^{N} v_{i} e_{ji}^{\circ} .
\end{align}
We are interested in the effects of selection on the mean frequency of $y$. Since $u\left(1/n\right)\sum_{i,j=1}^{N} v_{i}e_{ji}^{\circ}$ is independent of both $\vx$ and $\delta$, we have $\frac{d}{d\delta}\Big\vert_{\delta =0}\mathbb{E}_{\MSS}\left[u\left(1/n\right)\sum_{i,j=1}^{N} v_{i}e_{ji}^{\circ}\right] =0$, which gives
\begin{align}
-\frac{d}{d\delta}\Bigg\vert_{\delta =0}\mathbb{E}_{\MSS}\left[\Delta_{v,y}^{\circ}\left(\vX\right)\right] &= \frac{d}{d\delta}\Bigg\vert_{\delta =0} \mathbb{E}_{\MSS}\left[ \sum_{i,j=1}^{N}\left(v_{i} e_{ji}^{\circ}-\left(1-u\right) v_{j} e_{ij}^{\circ}\right) \mathbb{1}_{y}\left(X_{i}\right) \right] .
\end{align}
Since we are interested in $\frac{d}{d\delta}\Big\vert_{\delta =0}\left(1/N\right)\sum_{i=1}^{N}\mathbb{P}\left[X_{i}=y\right]$, to obtain an expression for the effects of weak selection on this mean frequency, it suffices to find a vector $v$ such that the equation
\begin{align}
\sum_{i,j=1}^{N}\left(v_{i} e_{ji}^{\circ}-\left(1-u\right) v_{j} e_{ij}^{\circ}\right) \mathbb{1}_{y}\left(x_{i}\right) &= \frac{1}{N}\sum_{i=1}^{N}\mathbb{1}_{y}\left(x_{i}\right) \label{eq:meanFrequency}
\end{align}
holds for every $\vx\in\left\{1,\dots ,n\right\}^{N}$ and $y\in\left\{1,\dots ,n\right\}$. In particular, we wish to find $v$ such that $\sum_{j=1}^{N}\left(v_{i} e_{ji}^{\circ}-\left(1-u\right) v_{j} e_{ij}^{\circ}\right) =1/N$ for all $i=1,\dots ,N$. Let $A$ be the transition probability matrix for the ancestral random walk, defined by $A_{ij}\coloneqq e_{ji}^{\circ}/d_{i}^{\circ}$, where $d_{i}^{\circ}\coloneqq\sum_{j=1}^{N}e_{ji}^{\circ}$ is the death probability of $i$ at neutral drift \citep{allen:JMB:2019}. In terms of ancestry, $A_{ij}$ is the probability that a lineage at location $i$ in the population now came from location $j$ in the previous time step. Expanding each summation over $j$ on the left-hand side of \eq{meanFrequency} gives
\begin{align}
\sum_{j=1}^{N}\left( v_{i} e_{ji}^{\circ}-\left(1-u\right) v_{j} e_{ij}^{\circ}\right) &= \sum_{j=1}^{N} \left(v_{j}d_{j}^{\circ}\right) \left(I-\left(1-u\right) A\right)_{ji} .
\end{align}
Since $A$ is stochastic and $u>0$, $I-\left(1-u\right) A$ must be invertible; otherwise, $1/\left(1-u\right)$ would be an eigenvalue of $A$, which is strictly greater than one. Therefore, the vector $v^{\ast}$ defined by
\begin{align}
v_{i}^{\ast} &\coloneqq \frac{1}{d_{i}^{\circ}}\frac{1}{N}\sum_{j=1}^{N}\left(\left(I-\left(1-u\right) A\right)^{-1}\right)_{ji}\label{eq:vasti}
\end{align}
is the unique solution to \eq{meanFrequency}. For this weighting, we have the equation
\begin{align}
-\frac{d}{d\delta}\Bigg\vert_{\delta =0}\mathbb{E}_{\MSS}\left[\Delta_{v^{\ast},y}^{\circ}\left(\vX\right)\right] &= \frac{d}{d\delta}\Bigg\vert_{\delta =0}\frac{1}{N}\sum_{i=1}^{N}\mathbb{P}_{\MSS}\left[X_{i}=y\right] . \label{eq:dDeltaNeutral}
\end{align}

\begin{remark}
Let $\varphi_{v,y}\left(\vx\right)\coloneqq\sum_{i=1}^{N}v_{i}\mathbb{1}_{y}\left(x_{i}\right)$ and $f_{y}\left(\vx\right)\coloneqq\left(1/N\right)\sum_{i=1}^{N}\mathbb{1}_{y}\left(x_{i}\right)$. To understand the first-order effects of selection on the mean frequency of $y$ from \eq{delprime}, we seek $v$ such that for some $C\in\mathbb{R}$, the discrete Laplacian of $\varphi_{v,y}$, defined by $\Delta\varphi_{v,y}\left(\vx\right)\coloneqq\Delta_{v,y}^{\circ}\left(\vx\right)$, satisfies
\begin{align}
\Delta\varphi_{v,y}\left(\vx\right) &= \frac{1}{N}\sum_{i=1}^{N}\mathbb{1}_{y}\left(x_{i}\right) + C = f_{y}\left(\vx\right) + C
\end{align}
A solution to this equation, $v^{\ast}$, is therefore a solution to a discrete version of Poisson's equation. The constant $C$ is immaterial since we are concerned with the marginal effects of selection.
\end{remark}

As noted in the main text, the quantity $v_{i}^{\ast}$ is a type of reproductive value \citep{fisher:OUP:1930,maciejewski:JTB:2014a} but one that depends on the mutation probability because it accounts for identity by descent. As noted there, we write $v_{i}^{\ast}=\pi_{i}^{\textrm{mut}}/u$ to maintain the analogy to a classical notion of reproductive value. As a result, from here on, we write our expressions in terms of $\pi^{\textrm{mut}}=uv^{\ast}$ instead of $v^{\ast}$ directly.

\subsubsection{Right-hand side of \eq{delprime}, $\mathbb{E}_{\MSS}^{\circ}\left[\frac{d}{d\delta}\Bigg\vert_{\delta =0}\Delta_{v,y}\left(\vX\right)\right]$}
The update rules we are interested in have the property that they depend on the state, $\vx$, through an intermediate state-to-fecundity map, $\vx\mapsto\mathbf{F}\left(\vx\right)$. The probability $p_{\left(R,\alpha\right)}\left(\vx\right)$ of replacement event $\left(R,\alpha\right)$ depends on $\vx$ only via the function $\mathbf{F}\left(\vx\right)$. We will assume that the fecundity of $i$ depends on its payoff, $U_{i}$, according to the map $F_{i}\left(\vx\right) =\exp\left\{\delta U_{i}\left(\vx\right)\right\}$, where $\delta\geqslant 0$ is the selection intensity. For weak selection, this mapping is equivalent to other commonly used payoff-to-fecundity mappings such as $F_{i}\left(\vx\right) =1+\delta U_{i}\left(\vx\right)$ \citep{maciejewski:PLoSCB:2014}. In particular, $\mathbf{F}^{\circ}=\mathbf{1}$ when $\delta =0$.

Let $m_{k}^{ij}\coloneqq\frac{\partial}{\partial F_{k}}\Big\vert_{\mathbf{F}=\mathbf{F}^{\circ}}e_{ij}\left(\mathbf{F}\right)$ be the marginal effect of small change in $k$'s fecundity, relative to its value under neutrality, on the probability of $i$ replacing $j$ \citep{mcavoy:NHB:2020}. The derivative of $e_{ij}$ at $\delta =0$ can thus be expressed in terms of these marginal effects and the payoff functions as follows:
\begin{align}
\frac{d}{d\delta}\Bigg\vert_{\delta =0} e_{ij}\left(\vx\right) &= \sum_{k=1}^{N} m_{k}^{ij} U_{k}\left(\vx\right) .
\end{align}
At $\delta =0$, the derivative of the expected change in the $v$-weighted abundance of $y$ is then
\begin{align}
\frac{d}{d\delta}\Bigg\vert_{\delta =0}\Delta_{v,y}\left(\vx\right) &= \frac{d}{d\delta}\Bigg\vert_{\delta =0}\sum_{i,j=1}^{N} v_{i} e_{ji}\left(\vx\right) \left( \left(1-u\right)\mathbb{1}_{y}\left(x_{j}\right) + u\frac{1}{n} - \mathbb{1}_{y}\left(x_{i}\right) \right) \nonumber \\
&= \sum_{i,j,k=1}^{N} v_{i} m_{k}^{ji} U_{k}\left(\vx\right) \left( \left(1-u\right)\mathbb{1}_{y}\left(x_{j}\right) + u\frac{1}{n} - \mathbb{1}_{y}\left(x_{i}\right) \right) . \label{eq:dDelta}
\end{align}

\subsubsection{Putting the left- and right-hand sides of \eq{delprime} together}
Combining \eq{dDeltaNeutral} and \eq{dDelta}, we see that
\begin{align}
\frac{d}{d\delta} &\Bigg\vert_{\delta =0}\frac{1}{N}\sum_{i=1}^{N}\mathbb{P}_{\MSS}\left[X_{i}=y\right] \nonumber \\
&= \mathbb{E}_{\MSS}^{\circ}\left[\sum_{i,j,k=1}^{N} v_{i}^{\ast} m_{k}^{ji} U_{k}\left(\vX\right) \left( \left(1-u\right)\mathbb{1}_{y}\left(X_{j}\right) + u\frac{1}{n} - \mathbb{1}_{y}\left(X_{i}\right) \right)\right] \nonumber \\
&= \sum_{i,j,k=1}^{N} v_{i}^{\ast} m_{k}^{ji} \left( \substack{\left(1-u\right) \mathbb{E}_{\MSS}^{\circ}\left[\mathbb{1}_{y}\left(X_{j}\right) U_{k}\left(\vX\right)\right] + u\frac{1}{n}\mathbb{E}_{\MSS}^{\circ}\left[U_{k}\left(\vX\right)\right] \\ - \mathbb{E}_{\MSS}^{\circ}\left[\mathbb{1}_{y}\left(X_{i}\right) U_{k}\left(\vX\right)\right]} \right) \nonumber \\
&= \sum_{i,j,k=1}^{N} v_{i}^{\ast} m_{k}^{ji} \left( \substack{\left(1-u\right) \frac{1}{n} \mathbb{E}_{\MSS}^{\circ}\left[U_{k}\left(\vX\right)\mid X_{j}=y\right] + u\frac{1}{n}\mathbb{E}_{\MSS}^{\circ}\left[U_{k}\left(\vX\right)\right] \\ - \frac{1}{n}\mathbb{E}_{\MSS}^{\circ}\left[U_{k}\left(\vX\right)\mid X_{i}=y\right]} \right) , \label{eq:frequencyDelta}
\end{align}
where we have used the fact that $\mathbb{P}_{\MSS}^{\circ}\left[X_{i}=y\right] =1/n$ for every $i=1,\dots ,N$. This gives \eq{xA_intuition2}.

\subsection{Mean strategy frequencies in matrix games}
Following \citet{tarnita:PNAS:2011}, our primary interest is in matrix games played between two players at a time. As the structure-coefficient theorem of  \citet{tarnita:PNAS:2011} applies only to symmetric games, we begin with these and reserve our extension of the theorem to asymmetric games until \textbf{\S\ref{sec:asymmetric}}. Let
\begin{align}
\bordermatrix{
&\ 1 &\ 2 &\ \cdots &\ n \cr
1\,\,\, &\ a_{11} &\ a_{12} &\ \cdots &\ a_{1n} \cr
2 &\ a_{21} &\ a_{22} &\ \cdots &\ a_{2n} \cr
\,\vdots &\ \vdots & \vdots & \ddots &\ \vdots \cr
n &\ a_{n1} &\ a_{n2} &\ \cdots &\ a_{nn}
}\label{eq:symmetricPayoffMatrix}
\end{align}
be such a payoff matrix. An individual using strategy $x$ against an opponent using $y$ gets $a_{xy}$ (and the opponent gets $a_{yx}$). These payoffs are averaged over partners to give a total payoff 
\begin{align}
U_{i}\left(\vx\right) &= \sum_{\ell =1}^{N} \Omega_{i\ell} \sum_{y_{i},y_{\ell}=1}^{n} \mathbb{1}_{y_{i}}\left(x_{i}\right)\mathbb{1}_{y_{\ell}}\left(x_{\ell}\right) a_{y_{i}y_{\ell}} , \label{eq:Ui_symmetric}
\end{align}
where $i$ is the individual in question and $\left(\Omega_{ij}\right)_{i,j=1}^{N}$ is an interaction matrix.

\begin{example}[Games on graphs]
Suppose that $\left(w_{ij}\right)_{i,j=1}^{N}$ is the adjacency matrix for an undirected, unweighted graph with $N$ vertices. If games are played with neighbors based on \eq{symmetricPayoffMatrix}, then the regime of accumulated payoffs is obtained by letting $\Omega_{ij}=w_{ij}$. The regime of averaged payoffs can be recovered by letting $\Omega_{ij}=w_{ij}/w_{i}$, where $w_{i}=\sum_{k=1}^{N}w_{ik}$ is the degree of vertex $i$.
\end{example}

In this setting, it follows from \eq{frequencyDelta} that
\begin{align}
\frac{d}{d\delta}&\Bigg\vert_{\delta =0}\frac{1}{N} \sum_{i=1}^{N}\mathbb{P}_{\MSS}\left[X_{i}=y\right] \nonumber \\
&= \sum_{i,j,k=1}^{N} v_{i}^{\ast} m_{k}^{ji} \sum_{\ell =1}^{N} \Omega_{k\ell} \sum_{y_{k},y_{\ell}=1}^{n} \left( \substack{\left(1-u\right)\mathbb{E}_{\MSS}^{\circ}\left[\mathbb{1}_{y}\left(X_{j}\right)\mathbb{1}_{y_{k}}\left(X_{k}\right)\mathbb{1}_{y_{\ell}}\left(X_{\ell}\right)\right] + u\frac{1}{n}\mathbb{E}_{\MSS}^{\circ}\left[\mathbb{1}_{y_{k}}\left(X_{k}\right)\mathbb{1}_{y_{\ell}}\left(X_{\ell}\right)\right] \\ - \mathbb{E}_{\MSS}^{\circ}\left[\mathbb{1}_{y}\left(X_{i}\right)\mathbb{1}_{y_{k}}\left(X_{k}\right)\mathbb{1}_{y_{\ell}}\left(X_{\ell}\right)\right]} \right) a_{y_{k}y_{\ell}} . \label{eq:mainBeforeJoint}
\end{align}
From \eq{mainBeforeJoint}, we see that we need to know the joint strategy distribution for up to three distinct individuals in the neutral MSS distribution. Once this joint distribution is known for all triples of individuals, we can use \eq{mainBeforeJoint} to calculate the first-order effects of selection on the mean frequency of strategy $y$ in the population. We calculate this joint distribution in the next section.

\section{Neutral drift}\label{sec:neutral_drift}
Due the symmetry of the mutation process, the stationary distribution, $\pi_{\MSS}^{\circ}$, has the following permutation invariance (which is just a slightly more refined version of the symmetry argument given by \citet{tarnita:PNAS:2011}):
\begin{lemma}\label{lem:symmetry}
Let $\tau :\left\{1,\dots ,n\right\}\rightarrow\left\{1,\dots ,n\right\}$ be a bijection. For each $\vx\in\left\{1,\dots ,n\right\}^{N}$, we have
\begin{align}
\pi_{\MSS}^{\circ}\left( \tau\left(\vx\right) \right) &= \pi_{\MSS}^{\circ}\left( \vx \right) ,
\end{align}
where $\tau\left(\vx\right)$ is the state obtained from $\vx$ by applying $\tau$ coordinatewise, i.e. $\tau\left(\vx\right)_{i}=\tau\left(x_{i}\right)$.
\end{lemma}
\begin{proof}
Let $P_{\vx\rightarrow\vy}^{\circ}$ denote the probability of transitioning from $\vx$ to $\vy$ in one time step, under neutral drift. Since $\pi_{\MSS}^{\circ}$ is a stationary distribution for this chain, for each $\vx\in\left\{1,\dots ,n\right\}^{N}$ we have
\begin{align}
\pi_{\MSS}^{\circ}\left( \tau\left(\vx\right) \right) &= \mathbb{E}_{\vy\sim\pi_{\MSS}^{\circ}} \left[ P_{\vy\rightarrow \tau\left(\vx\right)}^{\circ}\right] \nonumber \\
&= \mathbb{E}_{\vy\sim\pi_{\MSS}^{\circ}} \left[ \mathbb{E}_{\left(R,\alpha\right)}^{\circ} \left[ \substack{\prod_{i\in R}\left(\left(1-u\right)\mathbb{1}_{\tau^{-1}\left(y_{\alpha\left(i\right)}\right)}\left(x_{i}\right) +\frac{u}{n}\right) \vspace{0.2cm} \\ \times \prod_{i\not\in R}\mathbb{1}_{\tau^{-1}\left(y_{i}\right)}\left(x_{i}\right)} \right] \right] \nonumber \\
&= \mathbb{E}_{\vy\sim\pi_{\MSS}^{\circ}} \left[ \mathbb{E}_{\left(R,\alpha\right)}^{\circ} \left[ \substack{\prod_{i\in R}\left(\left(1-u\right)\mathbb{1}_{\tau^{-1}\left(\vy\right)_{\alpha\left(i\right)}}\left(x_{i}\right) +\frac{u}{n}\right) \vspace{0.2cm} \\ \times \prod_{i\not\in R}\mathbb{1}_{\tau^{-1}\left(\vy\right)_{i}}\left(x_{i}\right)} \right] \right] \nonumber \\
&= \mathbb{E}_{\tau\left(\vy\right)\sim\pi_{\MSS}^{\circ}} \left[ \mathbb{E}_{\left(R,\alpha\right)}^{\circ} \left[ \substack{\prod_{i\in R}\left(\left(1-u\right)\mathbb{1}_{y_{\alpha\left(i\right)}}\left(x_{i}\right) +\frac{u}{n}\right) \vspace{0.2cm} \\ \times \prod_{i\not\in R}\mathbb{1}_{y_{i}}\left(x_{i}\right)} \right] \right] \nonumber \\
&= \mathbb{E}_{\tau\left(\vy\right)\sim\pi_{\MSS}^{\circ}} \left[ P_{\vy\rightarrow\vx}^{\circ}\right] .
\end{align}
Thus, $\pi_{\MSS}^{\circ}\circ\tau$ is also a stationary distribution, so it must be equal to $\pi_{\MSS}^{\circ}$ by uniqueness.
\end{proof}

In our analysis, we need to know the joint strategy distribution for up to three distinct individuals, $i$, $j$, and $k$, in the neutral MSS distribution. By \lem{symmetry}, we can reduce this problem to whether pairs of individuals have the same strategy or different strategies. For example, let $\phi_{ij}$ be the probability that $i$ and $j$ have the same type in the neutral MSS distribution. The probability that $i$ and $j$ both have type $x\in\left\{1,\dots ,n\right\}$ is then $\phi_{ij}/n$. The probability that $i$ has type $x$ and $j$ has type $y\neq x$ is $\left(1-\phi_{ij}\right) /\left(n\left(n-1\right)\right)$. Thus, the joint distribution of $i$ and $j$ can be characterized in terms just $\phi_{ij}$. When $u=1$, strategies are chosen independently and uniformly-at-random from $\left\{1,\dots ,n\right\}$, so $\phi_{ij}=1/n$ for all $i,j\in\left\{1,\dots ,N\right\}$ with $i\neq j$. When $u\in\left(0,1\right)$, we have:
\begin{proposition}\label{prop:phi_ij}
The terms $\left\{\phi_{ij}\right\}_{i,j=1}^{N}$ are uniquely defined by $\phi_{ij}=1$ when $i=j$ and, when $i\neq j$,
\begin{align}
\phi_{ij} &= \underbrace{u^{2}\frac{1}{n} \sum_{\substack{\left(R,\alpha\right) \\ \left|\left\{i,j\right\}\cap R\right| =2}} p_{\left(R,\alpha\right)}^{\circ}}_{2\textrm{ mutations into }\left\{i,j\right\}} + \underbrace{u\frac{1}{n}\sum_{\substack{\left(R,\alpha\right) \\ \left|\left\{i,j\right\}\cap R\right| =1}} p_{\left(R,\alpha\right)}^{\circ} + 2u\left(1-u\right)\frac{1}{n}\sum_{\substack{\left(R,\alpha\right) \\ \left|\left\{i,j\right\}\cap R\right| =2}} p_{\left(R,\alpha\right)}^{\circ}}_{1\textrm{ mutation into }\left\{i,j\right\}} \nonumber \\[0.5cm]
&\qquad + \underbrace{\sum_{\left(R,\alpha\right)} p_{\left(R,\alpha\right)}^{\circ} \left(1-u\right)^{\left|\left\{i,j\right\}\cap R\right|} \phi_{\widetilde{\alpha}\left(i\right)\widetilde{\alpha}\left(j\right)}}_{0\textrm{ mutations into }\left\{i,j\right\}} . \label{eq:phi_ij_recurrence}
\end{align}
\end{proposition}

\begin{proof}
By definition, $\phi_{ii}=1$. For $i\neq j$, a one-step analysis of the neutral process gives
\begin{align}
\phi_{ij} &= \mathbb{P}^{\circ}\left[ \left|\left\{i,j\right\}\cap R\right| =0 \right] \phi_{ij} + \sum_{\substack{\left(R,\alpha\right) \\ \left\{i,j\right\}\cap R=\left\{i\right\}}} p_{\left(R,\alpha\right)}^{\circ} \Bigg[ \substack{\phi_{\alpha\left(i\right) j} \left(1-u+u\frac{1}{n}\right) \\ +\left(1-\phi_{\alpha\left(i\right) j}\right)\left(u\frac{1}{n}\right)} \Bigg] \nonumber \\
&\qquad + \sum_{\substack{\left(R,\alpha\right) \\ \left\{i,j\right\}\cap R=\left\{j\right\}}} p_{\left(R,\alpha\right)}^{\circ} \Bigg[ \substack{\phi_{i\alpha\left(j\right)} \left(1-u+u\frac{1}{n}\right) \\ +\left(1-\phi_{i\alpha\left(j\right)}\right)\left(u\frac{1}{n}\right)} \Bigg] \nonumber \\
&\qquad + \sum_{\substack{\left(R,\alpha\right) \\ \left|\left\{i,j\right\}\cap R\right| =2}} p_{\left(R,\alpha\right)}^{\circ} \Bigg[ \substack{\phi_{\alpha\left(i\right)\alpha\left(j\right)} \left(\left(1-u\right)^{2}+2u\left(1-u\right)\frac{1}{n} + u^{2}\frac{1}{n}\right) \\ +\left(1-\phi_{\alpha\left(i\right)\alpha\left(j\right)}\right) \left(2u\left(1-u\right)\frac{1}{n}+u^{2}\frac{1}{n}\right)} \Bigg] . \label{eq:phi_ij_oneStep}
\end{align}
To get \eq{phi_ij_recurrence}, we decompose \eq{phi_ij_oneStep} based on the number of mutations leading to each term. More specifically, since the number of individuals in $\left\{i,j\right\}$ replaced in one step of the process is between zero and two (inclusive), we consider at most two mutations arising in $\left\{i,j\right\}$.

The portion of \eq{phi_ij_oneStep} associated to exactly two mutations into $\left\{i,j\right\}$ is
\begin{align}
\sum_{\substack{\left(R,\alpha\right) \\ \left|\left\{i,j\right\}\cap R\right| =2}} p_{\left(R,\alpha\right)}^{\circ} \Bigg[ \substack{\phi_{\alpha\left(i\right)\alpha\left(j\right)} \left(u^{2}\frac{1}{n}\right) \\ +\left(1-\phi_{\alpha\left(i\right)\alpha\left(j\right)}\right) \left(u^{2}\frac{1}{n}\right)} \Bigg] &= u^{2}\frac{1}{n}\sum_{\substack{\left(R,\alpha\right) \\ \left|\left\{i,j\right\}\cap R\right| =2}} p_{\left(R,\alpha\right)}^{\circ} .
\end{align}
If both $i$ and $j$ are replaced, then mutations occur in both locations with probability $u^{2}$. Since the type of a mutation is chosen uniformly at random from the two available strategies, the probability that $i$ and $j$ are identical by state after both mutations is $1/n$.

The portion of \eq{phi_ij_oneStep} associated to exactly one mutation into $\left\{i,j\right\}$ is
\begin{align}
\sum_{\substack{\left(R,\alpha\right) \left\{i,j\right\}\cap R=\left\{i\right\}}} &p_{\left(R,\alpha\right)}^{\circ} \Bigg[ \substack{\phi_{\alpha\left(i\right) j} \left(u\frac{1}{n}\right) \\ +\left(1-\phi_{\alpha\left(i\right) j}\right)\left(u\frac{1}{n}\right)} \Bigg] + \sum_{\substack{\left(R,\alpha\right) \\ \left\{i,j\right\}\cap R=\left\{j\right\}}} p_{\left(R,\alpha\right)}^{\circ} \Bigg[ \substack{\phi_{i\alpha\left(j\right)} \left(u\frac{1}{n}\right) \\ +\left(1-\phi_{i\alpha\left(j\right)}\right)\left(u\frac{1}{n}\right)} \Bigg] \nonumber \\
&\qquad + \sum_{\substack{\left(R,\alpha\right) \\ \left|\left\{i,j\right\}\cap R\right| =2}} p_{\left(R,\alpha\right)}^{\circ} \Bigg[ \substack{\phi_{\alpha\left(i\right)\alpha\left(j\right)} \left(2u\left(1-u\right)\frac{1}{n}\right) \\ +\left(1-\phi_{\alpha\left(i\right)\alpha\left(j\right)}\right) \left(2u\left(1-u\right)\frac{1}{n}\right)} \Bigg] \nonumber \\
&= u\frac{1}{n}\sum_{\substack{\left(R,\alpha\right) \\ \left|\left\{i,j\right\}\cap R\right| =1}} p_{\left(R,\alpha\right)}^{\circ} + 2u\left(1-u\right)\frac{1}{n}\sum_{\substack{\left(R,\alpha\right) \\ \left|\left\{i,j\right\}\cap R\right| =2}} p_{\left(R,\alpha\right)}^{\circ} .
\end{align}
If exactly one of $i$ and $j$ is replaced, then there can be at most one mutation into the set $\left\{i,j\right\}$. The probability of one mutation is $u$, in which case the mutated type matches the type of the unreplaced individual with probability $1/n$. If both of $i$ and $j$ are replaced, then the probability that a mutation occurs at exactly one of $i$ and $j$ is $2u\left(1-u\right)$. The probability that the mutated type is that of the inherited type (at the other location) is $1/n$.

Finally, the portion of \eq{phi_ij_oneStep} associated to zero mutations into $\left\{i,j\right\}$ is
\begin{align}
\sum_{\left(R,\alpha\right)} p_{\left(R,\alpha\right)}^{\circ} \left(1-u\right)^{\left|\left\{i,j\right\}\cap R\right|} \phi_{\widetilde{\alpha}\left(i\right)\widetilde{\alpha}\left(j\right)} .
\end{align}
If neither $i$ nor $j$ is replaced, then no mutations into $\left\{i,j\right\}$ are possible, so $i$ and $j$ are identical by state in the next step if and only if they are identical by state in the previous step. If exactly one of $i$ and $j$ is replaced, then the probability of zero mutations into $\left\{i,j\right\}$ is $1-u$. In this case, if $i$ (resp. $j$) is replaced, then $i$ and $j$ are identical by state if and only if $\alpha\left(i\right)$ and $j$ (resp. $i$ and $\alpha\left(j\right)$) are identical by state. And if both of $i$ and $j$ are replaced, then the probability of zero mutations into $\left\{i,j\right\}$ is $\left(1-u\right)^{2}$, and $i$ and $j$ are identical by state if and only if $\alpha\left(i\right)$ and $\alpha\left(j\right)$ are identical by state.

Collecting all of these terms then gives \eq{phi_ij_recurrence}, as desired.
\end{proof}
Therefore, we can determine $\left\{\phi_{ij}\right\}_{i,j=1}^{N}$ by solving a linear system with $O\left(N^{2}\right)$ terms.

For $i\neq j$, there are just two possibilities: $i$ and $j$ have the same strategy or they do not. However, if there are $n\geqslant 3$ strategies, then for $i,j,k\in\left\{1,\dots ,N\right\}$ with $i\neq j\neq k\neq i$, there are five equivalence classes: \emph{(i)} $x_{i}=x_{j}=x_{k}$, \emph{(ii)} $x_{i}=x_{j}\neq x_{k}$, \emph{(iii)} $x_{i}=x_{k}\neq x_{j}$, \emph{(iv)} $x_{i}\neq x_{j}=x_{k}$, and \emph{(v)} $x_{i}\neq x_{j}\neq x_{k}\neq x_{i}$. Using the subscript $L\in\left\{1,2,3,4,5\right\}$ to denote the equivalence class, let $\psi_{ijk}^{\left(L\right)}$ be the probability that $\left(X_{i},X_{j},X_{k}\right)$ belongs to class $\mathcal{C}^{\left(L\right)}$ in the neutral MSS distribution,
\begin{align}
\psi_{ijk}^{\left(L\right)} &= \sum_{\left(x_{i},x_{j},x_{k}\right)\in\mathcal{C}^{\left(L\right)}} \mathbb{P}_{\MSS}^{\circ}\left[ X_{i}=x_{i}, X_{j}=x_{j}, X_{k}=x_{k} \right] .
\end{align}
The size of each of these equivalence classes are as follows:
\begin{subequations}
\begin{align}
\left|\mathcal{C}^{\left(1\right)}\right| &= n ; \\
\left|\mathcal{C}^{\left(2\right)}\right| = \left|\mathcal{C}^{\left(3\right)}\right| = \left|\mathcal{C}^{\left(4\right)}\right| &= n\left(n-1\right) ; \\
\left|\mathcal{C}^{\left(5\right)}\right| &= n\left(n-1\right)\left(n-2\right) .
\end{align}
\end{subequations}
For example, if $x\neq y\neq z\neq x$, then the probability that $i$, $j$, and $k$ have strategies $x$, $y$, and $z$ in the neutral MSS distribution is $\psi_{ijk}^{\left(5\right)}/\left(n\left(n-1\right)\left(n-2\right)\right)$. The joint distribution of $i$, $j$, and $k$ can thus be described by at most five quantities, $\left\{\psi_{ijk}^{\left(L\right)}\right\}_{L=1}^{5}$. In fact, we can reduce it to just $\psi_{ijk}^{\left(1\right)}$:
\begin{lemma}\label{lem:psi_ijk}
For $L=2,3,4,5$, we can express $\psi_{ijk}^{\left(L\right)}$ in terms of $\psi_{ijk}^{\left(1\right)}$ as follows:
\begin{subequations}\label{eq:psi_ijk}
\begin{align}
\psi_{ijk}^{\left(2\right)} &= \phi_{ij}-\psi_{ijk}^{\left(1\right)} ; \\
\psi_{ijk}^{\left(3\right)} &= \phi_{ik}-\psi_{ijk}^{\left(1\right)} ; \\
\psi_{ijk}^{\left(4\right)} &= \phi_{jk}-\psi_{ijk}^{\left(1\right)} ; \\
\psi_{ijk}^{\left(5\right)} &= 1+2\psi_{ijk}^{\left(1\right)}-\phi_{ij}-\phi_{ik}-\phi_{jk} .
\end{align}
\end{subequations}
\end{lemma}
\begin{proof}[Proof (sketch)]
The first three equations are true more-or-less by definition, e.g.
\begin{align}
\psi_{ijk}^{\left(2\right)} &= \sum_{x=1}^{n} \sum_{y\neq x} \mathbb{P}_{\MSS}^{\circ}\left[ X_{i}=X_{j}=x, X_{k}=y \right] \nonumber \\
&= \sum_{x=1}^{n} \mathbb{P}_{\MSS}^{\circ}\left[ X_{i}=X_{j}=x, X_{k}\in\left\{1,\dots ,n\right\}\right] - \sum_{x=1}^{n} \mathbb{P}_{\MSS}^{\circ}\left[ X_{i}=X_{j}=X_{k}=x \right] \nonumber \\
&= \sum_{x=1}^{n} \mathbb{P}_{\MSS}^{\circ}\left[ X_{i}=X_{j}=x\right] - \sum_{x=1}^{n} \mathbb{P}_{\MSS}^{\circ}\left[ X_{i}=X_{j}=X_{k}=x \right] \nonumber \\
&= \phi_{ij}-\psi_{ijk}^{\left(1\right)} .
\end{align}
The last equation can be derived from the first three by noting that
\begin{align}
\psi_{ijk}^{\left(1\right)} + \psi_{ijk}^{\left(2\right)} + \psi_{ijk}^{\left(3\right)} + \psi_{ijk}^{\left(4\right)} + \psi_{ijk}^{\left(5\right)} &= 1 ,
\end{align}
which gives the desired expressions for $\psi_{ijk}^{\left(L\right)}$ ($L=2,3,4,5$) in terms of $\psi_{ijk}^{\left(1\right)}$.
\end{proof}

To simplify notation, we let $\phi_{ijk}\coloneqq\psi_{ijk}^{\left(1\right)}$. If $n=2$, then $\phi_{ijk}=\left(\phi_{ij}+\phi_{ik}+\phi_{jk}-1\right) /2$ \citep[][Equation 103]{allen:JMB:2019}, so calculating $\left\{\phi_{ijk}\right\}_{i,j,k=1}^{N}$ is trivial once we know $\left\{\phi_{ij}\right\}_{i,j=1}^{N}$. However, no such formula can hold generically when there are $n\geqslant 3$ strategies, as the following lemma shows:
\begin{lemma}\label{lem:no_triplet_to_pair}
For $n\geqslant 3$, there does not exist an expression of the form
\begin{align}
\phi_{ijk} &= c_{0}\left(n\right) +c_{1}\left(n\right)\phi_{ij}+c_{2}\left(n\right)\phi_{ik}+c_{3}\left(n\right)\phi_{jk}
\end{align}
that holds generally for every $i,j,k\in\left\{1,\dots ,N\right\}$.
\end{lemma}
\begin{proof}
Suppose that $\phi_{ijk}=c_{0}\left(n\right) +c_{1}\left(n\right)\phi_{ij}+c_{2}\left(n\right)\phi_{ik}+c_{3}\left(n\right)\phi_{jk}$. By symmetry, it must be true that $c_{1}\left(n\right) =c_{2}\left(n\right) =c_{3}\left(n\right)$. Moreover, since this equation must hold for every $i,j,k\in\left\{1,\dots ,N\right\}$, letting $i=j=k$ gives $c_{0}\left(n\right) +3c_{1}\left(n\right) =1$. If $i\neq j=k$, then $c_{0}\left(n\right) +c_{1}\left(n\right)\left(2\phi_{ij}+1\right) =\phi_{ij}$. Solving this system of equations yields $c_{0}\left(n\right) =-1/2$ and $c_{1}\left(n\right) =1/2$. By \eq{psi_ijk}, we then see that $\psi_{ijk}^{\left(5\right)}=0$, which is true in general only for $n=2$ strategies. For $n\geqslant 3$, there is (generically) a nonzero probability that $X_{i}\neq X_{j}\neq X_{k}\neq X_{i}$.
\end{proof}

Therefore, unlike the case $n=2$, here we see no way to reduce $\left\{\phi_{ijk}\right\}_{i,j,k=1}^{N}$ to $\left\{\phi_{ij}\right\}_{i,j=1}^{N}$ when $n\geqslant 3$. Instead, $\left\{\phi_{ijk}\right\}_{i,j,k=1}^{N}$ can be calculated by solving a recurrence relation:
\begin{proposition}\label{prop:phi_ijk}
If $i\neq j\neq k\neq i$, then $\phi_{ijk}$ satisfies
\begin{align}
\phi_{ijk} &= \underbrace{u^{3}\frac{1}{n^{2}}\sum_{\substack{\left(R,\alpha\right) \\ \left|\left\{i,j,k\right\}\cap R\right| =3}} p_{\left(R,\alpha\right)}^{\circ}}_{3\textrm{ mutations into }\left\{i,j,k\right\}} + \underbrace{u^{2}\frac{1}{n^{2}}\sum_{\substack{\left(R,\alpha\right) \\ \left|\left\{i,j,k\right\}\cap R\right| =2}} p_{\left(R,\alpha\right)}^{\circ} + 3u^{2}\left(1-u\right)\frac{1}{n^{2}}\sum_{\substack{\left(R,\alpha\right) \\ \left|\left\{i,j,k\right\}\cap R\right| =3}} p_{\left(R,\alpha\right)}^{\circ}}_{2\textrm{ mutations into }\left\{i,j,k\right\}} \nonumber \\[0.5cm]
&\qquad + u\frac{1}{n}\phi_{jk}\sum_{\substack{\left(R,\alpha\right) \\ \left\{i,j,k\right\}\cap R=\left\{i\right\}}} p_{\left(R,\alpha\right)}^{\circ} + u\frac{1}{n}\phi_{ik}\sum_{\substack{\left(R,\alpha\right) \\ \left\{i,j,k\right\}\cap R=\left\{j\right\}}} p_{\left(R,\alpha\right)}^{\circ} + u\frac{1}{n}\phi_{ij}\sum_{\substack{\left(R,\alpha\right) \\ \left\{i,j,k\right\}\cap R=\left\{k\right\}}} p_{\left(R,\alpha\right)}^{\circ} \nonumber \\
&\qquad + u\left(1-u\right)\frac{1}{n}\sum_{\substack{\left(R,\alpha\right) \\ \left\{i,j,k\right\}\cap R=\left\{j,k\right\}}} p_{\left(R,\alpha\right)}^{\circ} \left(\phi_{i\alpha\left(j\right)}+\phi_{i\alpha\left(k\right)}\right) \nonumber \\
&\qquad + u\left(1-u\right)\frac{1}{n}\sum_{\substack{\left(R,\alpha\right) \\ \left\{i,j,k\right\}\cap R=\left\{i,k\right\}}} p_{\left(R,\alpha\right)}^{\circ} \left(\phi_{\alpha\left(i\right) j}+\phi_{j\alpha\left(k\right)}\right) \nonumber \\
&\qquad + u\left(1-u\right)\frac{1}{n}\sum_{\substack{\left(R,\alpha\right) \\ \left\{i,j,k\right\}\cap R=\left\{i,j\right\}}} p_{\left(R,\alpha\right)}^{\circ} \left(\phi_{\alpha\left(i\right) k}+\phi_{\alpha\left(j\right) k}\right) \nonumber \\
&\qquad \underbrace{+\hspace{0.5ex} u\left(1-u\right)^{2}\frac{1}{n}\sum_{\substack{\left(R,\alpha\right) \\ \left|\left\{i,j,k\right\}\cap R\right| =3}} p_{\left(R,\alpha\right)}^{\circ} \left(\phi_{\alpha\left(i\right)\alpha\left(j\right)}+\phi_{\alpha\left(i\right)\alpha\left(k\right)}+\phi_{\alpha\left(j\right)\alpha\left(k\right)}\right)\hspace{3cm}}_{1\textrm{ mutation into }\left\{i,j,k\right\}} \nonumber \\[0.5cm]
&\qquad + \underbrace{\sum_{\left(R,\alpha\right)} p_{\left(R,\alpha\right)}^{\circ} \left(1-u\right)^{\left|\left\{i,j,k\right\}\cap R\right|} \phi_{\widetilde{\alpha}\left(i\right)\widetilde{\alpha}\left(j\right)\widetilde{\alpha}\left(k\right)}}_{0\textrm{ mutations into }\left\{i,j,k\right\}} . \label{eq:phi_ijk_recurrence}
\end{align}
The terms in \eq{phi_ijk_recurrence} involving fewer than three distinct individuals can be calculated using \prop{phi_ij}. Thus, \eq{phi_ijk_recurrence} together with \prop{phi_ij} uniquely define $\left\{\phi_{ijk}\right\}_{i,j,k=1}^{N}$.
\end{proposition}
\begin{proof}
For $i\neq j\neq k\neq i$, a one-step analysis of the neutral process gives
\begin{align}
\phi_{ijk} &= \mathbb{P}^{\circ}\left[\left|\left\{i,j,k\right\}\cap R\right| =0\right] \phi_{ijk} + \sum_{\substack{\left(R,\alpha\right) \\ \left\{i,j,k\right\}\cap R=\left\{i\right\}}} p_{\left(R,\alpha\right)}^{\circ} \left[\substack{\psi_{\alpha\left(i\right) jk}^{\left(1\right)}\left(1-u+u\frac{1}{n}\right) \\ +\psi_{\alpha\left(i\right) jk}^{\left(4\right)}\left(u\frac{1}{n}\right)}\right] \nonumber \\
&\qquad + \sum_{\substack{\left(R,\alpha\right) \\ \left\{i,j,k\right\}\cap R=\left\{j\right\}}} p_{\left(R,\alpha\right)}^{\circ} \left[\substack{\psi_{i\alpha\left(j\right) k}^{\left(1\right)}\left(1-u+u\frac{1}{n}\right) \\ +\psi_{i\alpha\left(j\right) k}^{\left(3\right)}\left(u\frac{1}{n}\right)}\right] + \sum_{\substack{\left(R,\alpha\right) \\ \left\{i,j,k\right\}\cap R=\left\{k\right\}}} p_{\left(R,\alpha\right)}^{\circ} \left[\substack{\psi_{ij\alpha\left(k\right)}^{\left(1\right)}\left(1-u+u\frac{1}{n}\right) \\ +\psi_{ij\alpha\left(k\right)}^{\left(2\right)}\left(u\frac{1}{n}\right)}\right] \nonumber \\
&\qquad + \sum_{\substack{\left(R,\alpha\right) \\ \left\{i,j,k\right\}\cap R=\left\{j,k\right\}}} p_{\left(R,\alpha\right)}^{\circ} \left[\substack{\psi_{i\alpha\left(j\right)\alpha\left(k\right)}^{\left(1\right)}\left(\left(1-u\right)^{2}+2u\left(1-u\right)\frac{1}{n}+u^{2}\frac{1}{n^{2}}\right) \\ +\left(\psi_{i\alpha\left(j\right)\alpha\left(k\right)}^{\left(2\right)}+\psi_{i\alpha\left(j\right)\alpha\left(k\right)}^{\left(3\right)}\right)\left(u\left(1-u\right)\frac{1}{n}+u^{2}\frac{1}{n^{2}}\right) \\ +\left(\psi_{i\alpha\left(j\right)\alpha\left(k\right)}^{\left(4\right)}+\psi_{i\alpha\left(j\right)\alpha\left(k\right)}^{\left(5\right)}\right)\left(u^{2}\frac{1}{n^{2}}\right)}\right] \nonumber \\
&\qquad + \sum_{\substack{\left(R,\alpha\right) \\ \left\{i,j,k\right\}\cap R=\left\{i,k\right\}}} p_{\left(R,\alpha\right)}^{\circ} \left[\substack{\psi_{\alpha\left(i\right) j\alpha\left(k\right)}^{\left(1\right)}\left(\left(1-u\right)^{2}+2u\left(1-u\right)\frac{1}{n}+u^{2}\frac{1}{n^{2}}\right) \\ +\left(\psi_{\alpha\left(i\right) j\alpha\left(k\right)}^{\left(2\right)}+\psi_{\alpha\left(i\right) j\alpha\left(k\right)}^{\left(4\right)}\right)\left(u\left(1-u\right)\frac{1}{n}+u^{2}\frac{1}{n^{2}}\right) \\ +\left(\psi_{\alpha\left(i\right) j\alpha\left(k\right)}^{\left(3\right)}+\psi_{\alpha\left(i\right) j\alpha\left(k\right)}^{\left(5\right)}\right)\left(u^{2}\frac{1}{n^{2}}\right)}\right] \nonumber \\
&\qquad + \sum_{\substack{\left(R,\alpha\right) \\ \left\{i,j,k\right\}\cap R=\left\{i,j\right\}}} p_{\left(R,\alpha\right)}^{\circ} \left[\substack{\psi_{\alpha\left(i\right)\alpha\left(j\right) k}^{\left(1\right)}\left(\left(1-u\right)^{2}+2u\left(1-u\right)\frac{1}{n}+u^{2}\frac{1}{n^{2}}\right) \\ +\left(\psi_{\alpha\left(i\right)\alpha\left(j\right) k}^{\left(3\right)}+\psi_{\alpha\left(i\right)\alpha\left(j\right) k}^{\left(4\right)}\right)\left(u\left(1-u\right)\frac{1}{n}+u^{2}\frac{1}{n^{2}}\right) \\ +\left(\psi_{\alpha\left(i\right)\alpha\left(j\right) k}^{\left(2\right)}+\psi_{\alpha\left(i\right)\alpha\left(j\right) k}^{\left(5\right)}\right)\left(u^{2}\frac{1}{n^{2}}\right)}\right] \nonumber \\
&\qquad + \sum_{\substack{\left(R,\alpha\right) \\ \left|\left\{i,j,k\right\}\cap R\right| =3}} p_{\left(R,\alpha\right)}^{\circ} \left[\substack{\psi_{\alpha\left(i\right)\alpha\left(j\right)\alpha\left(k\right)}^{\left(1\right)}\left(\left(1-u\right)^{3}+3u\left(1-u\right)^{2}\frac{1}{n}+3u^{2}\left(1-u\right)\frac{1}{n^{2}} +u^{3}\frac{1}{n^{2}}\right) \\ +\left(\psi_{\alpha\left(i\right)\alpha\left(j\right)\alpha\left(k\right)}^{\left(2\right)}+\psi_{\alpha\left(i\right)\alpha\left(j\right)\alpha\left(k\right)}^{\left(3\right)}+\psi_{\alpha\left(i\right)\alpha\left(j\right)\alpha\left(k\right)}^{\left(4\right)}\right)\left(u\left(1-u\right)^{2}\frac{1}{n}+3u^{2}\left(1-u\right)\frac{1}{n^{2}}+u^{3}\frac{1}{n^{2}}\right) \\ +\psi_{\alpha\left(i\right)\alpha\left(j\right)\alpha\left(k\right)}^{\left(5\right)}\left(3u^{2}\left(1-u\right)\frac{1}{n^{2}} +u^{3}\frac{1}{n^{2}}\right)}\right] . \label{eq:phi_ijk_oneStep}
\end{align}
Like what we did in the proof of \prop{phi_ij}, to obtain \eq{phi_ijk_recurrence} we decompose \eq{phi_ijk_oneStep} based on the number of mutations into $\left\{i,j,k\right\}$ (which is now between zero and three, inclusive).

The portion of \eq{phi_ijk_oneStep} associated to exactly three mutations into $\left\{i,j,k\right\}$ is
\begin{align}
\sum_{\substack{\left(R,\alpha\right) \\ \left|\left\{i,j,k\right\}\cap R\right| =3}} p_{\left(R,\alpha\right)}^{\circ} \left[\substack{\psi_{\alpha\left(i\right)\alpha\left(j\right)\alpha\left(k\right)}^{\left(1\right)}\left(u^{3}\frac{1}{n^{2}}\right) \\ +\left(\psi_{\alpha\left(i\right)\alpha\left(j\right)\alpha\left(k\right)}^{\left(2\right)}+\psi_{\alpha\left(i\right)\alpha\left(j\right)\alpha\left(k\right)}^{\left(3\right)}+\psi_{\alpha\left(i\right)\alpha\left(j\right)\alpha\left(k\right)}^{\left(4\right)}\right)\left(u^{3}\frac{1}{n^{2}}\right) \\ +\psi_{\alpha\left(i\right)\alpha\left(j\right)\alpha\left(k\right)}^{\left(5\right)}\left(u^{3}\frac{1}{n^{2}}\right)}\right] = u^{3}\frac{1}{n^{2}}\sum_{\substack{\left(R,\alpha\right) \\ \left|\left\{i,j,k\right\}\cap R\right| =3}} p_{\left(R,\alpha\right)}^{\circ} .
\end{align}
The portion of \eq{phi_ijk_oneStep} associated to exactly two mutations into $\left\{i,j,k\right\}$ is
\begin{align}
&\sum_{\substack{\left(R,\alpha\right) \\ \left\{i,j,k\right\}\cap R=\left\{j,k\right\}}} p_{\left(R,\alpha\right)}^{\circ} \left[\substack{\psi_{i\alpha\left(j\right)\alpha\left(k\right)}^{\left(1\right)}\left(u^{2}\frac{1}{n^{2}}\right) \\ +\left(\psi_{i\alpha\left(j\right)\alpha\left(k\right)}^{\left(2\right)}+\psi_{i\alpha\left(j\right)\alpha\left(k\right)}^{\left(3\right)}\right)\left(u^{2}\frac{1}{n^{2}}\right) \\ +\left(\psi_{i\alpha\left(j\right)\alpha\left(k\right)}^{\left(4\right)}+\psi_{i\alpha\left(j\right)\alpha\left(k\right)}^{\left(5\right)}\right)\left(u^{2}\frac{1}{n^{2}}\right)}\right] \nonumber \\
&\qquad + \sum_{\substack{\left(R,\alpha\right) \\ \left\{i,j,k\right\}\cap R=\left\{i,k\right\}}} p_{\left(R,\alpha\right)}^{\circ} \left[\substack{\psi_{\alpha\left(i\right) j\alpha\left(k\right)}^{\left(1\right)}\left(u^{2}\frac{1}{n^{2}}\right) \\ +\left(\psi_{\alpha\left(i\right) j\alpha\left(k\right)}^{\left(2\right)}+\psi_{\alpha\left(i\right) j\alpha\left(k\right)}^{\left(4\right)}\right)\left(u^{2}\frac{1}{n^{2}}\right) \\ +\left(\psi_{\alpha\left(i\right) j\alpha\left(k\right)}^{\left(3\right)}+\psi_{\alpha\left(i\right) j\alpha\left(k\right)}^{\left(5\right)}\right)\left(u^{2}\frac{1}{n^{2}}\right)}\right] \nonumber \\
&\qquad + \sum_{\substack{\left(R,\alpha\right) \\ \left\{i,j,k\right\}\cap R=\left\{i,j\right\}}} p_{\left(R,\alpha\right)}^{\circ} \left[\substack{\psi_{\alpha\left(i\right)\alpha\left(j\right) k}^{\left(1\right)}\left(u^{2}\frac{1}{n^{2}}\right) \\ +\left(\psi_{\alpha\left(i\right)\alpha\left(j\right) k}^{\left(3\right)}+\psi_{\alpha\left(i\right)\alpha\left(j\right) k}^{\left(4\right)}\right)\left(u^{2}\frac{1}{n^{2}}\right) \\ +\left(\psi_{\alpha\left(i\right)\alpha\left(j\right) k}^{\left(2\right)}+\psi_{\alpha\left(i\right)\alpha\left(j\right) k}^{\left(5\right)}\right)\left(u^{2}\frac{1}{n^{2}}\right)}\right] \nonumber \\
&\qquad + \sum_{\substack{\left(R,\alpha\right) \\ \left|\left\{i,j,k\right\}\cap R\right| =3}} p_{\left(R,\alpha\right)}^{\circ} \left[\substack{\psi_{\alpha\left(i\right)\alpha\left(j\right)\alpha\left(k\right)}^{\left(1\right)}\left(3u^{2}\left(1-u\right)\frac{1}{n^{2}}\right) \\ +\left(\psi_{\alpha\left(i\right)\alpha\left(j\right)\alpha\left(k\right)}^{\left(2\right)}+\psi_{\alpha\left(i\right)\alpha\left(j\right)\alpha\left(k\right)}^{\left(3\right)}+\psi_{\alpha\left(i\right)\alpha\left(j\right)\alpha\left(k\right)}^{\left(4\right)}\right)\left(3u^{2}\left(1-u\right)\frac{1}{n^{2}}\right) \\ +\psi_{\alpha\left(i\right)\alpha\left(j\right)\alpha\left(k\right)}^{\left(5\right)}\left(3u^{2}\left(1-u\right)\frac{1}{n^{2}}\right)}\right] \nonumber \\
&= u^{2}\frac{1}{n^{2}}\sum_{\substack{\left(R,\alpha\right) \\ \left|\left\{i,j,k\right\}\cap R\right| =2}} p_{\left(R,\alpha\right)}^{\circ} + 3u^{2}\left(1-u\right)\frac{1}{n^{2}}\sum_{\substack{\left(R,\alpha\right) \\ \left|\left\{i,j,k\right\}\cap R\right| =3}} p_{\left(R,\alpha\right)}^{\circ} .
\end{align}
From \lem{psi_ijk}, the portion of \eq{phi_ijk_oneStep} associated to exactly one mutation into $\left\{i,j,k\right\}$ is
\begin{align}
&\sum_{\substack{\left(R,\alpha\right) \\ \left\{i,j,k\right\}\cap R=\left\{i\right\}}} p_{\left(R,\alpha\right)}^{\circ} \left[\substack{\psi_{\alpha\left(i\right) jk}^{\left(1\right)}\left(u\frac{1}{n}\right) \\ +\psi_{\alpha\left(i\right) jk}^{\left(4\right)}\left(u\frac{1}{n}\right)}\right] + \sum_{\substack{\left(R,\alpha\right) \\ \left\{i,j,k\right\}\cap R=\left\{j\right\}}} p_{\left(R,\alpha\right)}^{\circ} \left[\substack{\psi_{i\alpha\left(j\right) k}^{\left(1\right)}\left(u\frac{1}{n}\right) \\ +\psi_{i\alpha\left(j\right) k}^{\left(3\right)}\left(u\frac{1}{n}\right)}\right] \nonumber \\
&\qquad + \sum_{\substack{\left(R,\alpha\right) \\ \left\{i,j,k\right\}\cap R=\left\{k\right\}}} p_{\left(R,\alpha\right)}^{\circ} \left[\substack{\psi_{ij\alpha\left(k\right)}^{\left(1\right)}\left(u\frac{1}{n}\right) \\ +\psi_{ij\alpha\left(k\right)}^{\left(2\right)}\left(u\frac{1}{n}\right)}\right] + \sum_{\substack{\left(R,\alpha\right) \\ \left\{i,j,k\right\}\cap R=\left\{j,k\right\}}} p_{\left(R,\alpha\right)}^{\circ} \left[\substack{\psi_{i\alpha\left(j\right)\alpha\left(k\right)}^{\left(1\right)}\left(2u\left(1-u\right)\frac{1}{n}\right) \\ +\left(\psi_{i\alpha\left(j\right)\alpha\left(k\right)}^{\left(2\right)}+\psi_{i\alpha\left(j\right)\alpha\left(k\right)}^{\left(3\right)}\right)\left(u\left(1-u\right)\frac{1}{n}\right)}\right] \nonumber \\
&\qquad + \sum_{\substack{\left(R,\alpha\right) \\ \left\{i,j,k\right\}\cap R=\left\{i,k\right\}}} p_{\left(R,\alpha\right)}^{\circ} \left[\substack{\psi_{\alpha\left(i\right) j\alpha\left(k\right)}^{\left(1\right)}\left(2u\left(1-u\right)\frac{1}{n}\right) \\ +\left(\psi_{\alpha\left(i\right) j\alpha\left(k\right)}^{\left(2\right)}+\psi_{\alpha\left(i\right) j\alpha\left(k\right)}^{\left(4\right)}\right)\left(u\left(1-u\right)\frac{1}{n}\right)}\right] \nonumber \\
&\qquad + \sum_{\substack{\left(R,\alpha\right) \\ \left\{i,j,k\right\}\cap R=\left\{i,j\right\}}} p_{\left(R,\alpha\right)}^{\circ} \left[\substack{\psi_{\alpha\left(i\right)\alpha\left(j\right) k}^{\left(1\right)}\left(2u\left(1-u\right)\frac{1}{n}\right) \\ +\left(\psi_{\alpha\left(i\right)\alpha\left(j\right) k}^{\left(3\right)}+\psi_{\alpha\left(i\right)\alpha\left(j\right) k}^{\left(4\right)}\right)\left(u\left(1-u\right)\frac{1}{n}\right)}\right] \nonumber \\
&\qquad + \sum_{\substack{\left(R,\alpha\right) \\ \left|\left\{i,j,k\right\}\cap R\right| =3}} p_{\left(R,\alpha\right)}^{\circ} \left[\substack{\psi_{\alpha\left(i\right)\alpha\left(j\right)\alpha\left(k\right)}^{\left(1\right)}\left(3u\left(1-u\right)^{2}\frac{1}{n}\right) \\ +\left(\psi_{\alpha\left(i\right)\alpha\left(j\right)\alpha\left(k\right)}^{\left(2\right)}+\psi_{\alpha\left(i\right)\alpha\left(j\right)\alpha\left(k\right)}^{\left(3\right)}+\psi_{\alpha\left(i\right)\alpha\left(j\right)\alpha\left(k\right)}^{\left(4\right)}\right)\left(u\left(1-u\right)^{2}\frac{1}{n}\right)}\right] \nonumber \\
&= u\frac{1}{n}\phi_{jk}\sum_{\substack{\left(R,\alpha\right) \\ \left\{i,j,k\right\}\cap R=\left\{i\right\}}} p_{\left(R,\alpha\right)}^{\circ} + u\frac{1}{n}\phi_{ik}\sum_{\substack{\left(R,\alpha\right) \\ \left\{i,j,k\right\}\cap R=\left\{j\right\}}} p_{\left(R,\alpha\right)}^{\circ} + u\frac{1}{n}\phi_{ij}\sum_{\substack{\left(R,\alpha\right) \\ \left\{i,j,k\right\}\cap R=\left\{k\right\}}} p_{\left(R,\alpha\right)}^{\circ} \nonumber \\
&\qquad + u\left(1-u\right)\frac{1}{n}\sum_{\substack{\left(R,\alpha\right) \\ \left\{i,j,k\right\}\cap R=\left\{j,k\right\}}} p_{\left(R,\alpha\right)}^{\circ} \left(\phi_{i\alpha\left(j\right)}+\phi_{i\alpha\left(k\right)}\right) \nonumber \\
&\qquad + u\left(1-u\right)\frac{1}{n}\sum_{\substack{\left(R,\alpha\right) \\ \left\{i,j,k\right\}\cap R=\left\{i,k\right\}}} p_{\left(R,\alpha\right)}^{\circ} \left(\phi_{\alpha\left(i\right) j}+\phi_{j\alpha\left(k\right)}\right) \nonumber \\
&\qquad + u\left(1-u\right)\frac{1}{n}\sum_{\substack{\left(R,\alpha\right) \\ \left\{i,j,k\right\}\cap R=\left\{i,j\right\}}} p_{\left(R,\alpha\right)}^{\circ} \left(\phi_{\alpha\left(i\right) k}+\phi_{\alpha\left(j\right) k}\right) \nonumber \\
&\qquad + u\left(1-u\right)^{2}\frac{1}{n}\sum_{\substack{\left(R,\alpha\right) \\ \left|\left\{i,j,k\right\}\cap R\right| =3}} p_{\left(R,\alpha\right)}^{\circ} \left(\phi_{\alpha\left(i\right)\alpha\left(j\right)}+\phi_{\alpha\left(i\right)\alpha\left(k\right)}+\phi_{\alpha\left(j\right)\alpha\left(k\right)}\right) .
\end{align}
Finally, the portion of \eq{phi_ijk_oneStep} associated to zero mutations into $\left\{i,j,k\right\}$ is
\begin{align}
\sum_{\left(R,\alpha\right)} p_{\left(R,\alpha\right)}^{\circ} \left(1-u\right)^{\left|\left\{i,j,k\right\}\cap R\right|} \phi_{\widetilde{\alpha}\left(i\right)\widetilde{\alpha}\left(j\right)\widetilde{\alpha}\left(k\right)} .
\end{align}
Putting everything together then gives \eq{phi_ijk_recurrence}, as desired.
\end{proof}
It follows that we can find $\left\{\phi_{ijk}\right\}_{i,j,k=1}^{N}$ by solving a linear system with $O\left(N^{3}\right)$ terms.

\begin{corollary}
If at most one individual is replaced in each time step, then, for $i\neq j\neq k\neq i$,
\begin{align}
\phi_{ijk} &= u\frac{1}{n}\phi_{jk}\sum_{\substack{\left(R,\alpha\right) \\ \left\{i,j,k\right\}\cap R=\left\{i\right\}}} p_{\left(R,\alpha\right)}^{\circ} + u\frac{1}{n}\phi_{ik}\sum_{\substack{\left(R,\alpha\right) \\ \left\{i,j,k\right\}\cap R=\left\{j\right\}}} p_{\left(R,\alpha\right)}^{\circ} + u\frac{1}{n}\phi_{ij}\sum_{\substack{\left(R,\alpha\right) \\ \left\{i,j,k\right\}\cap R=\left\{k\right\}}} p_{\left(R,\alpha\right)}^{\circ} \nonumber \\
&\qquad + \phi_{ijk} \sum_{\substack{\left(R,\alpha\right) \\ \left|\left\{i,j,k\right\}\cap R\right| =0}} p_{\left(R,\alpha\right)}^{\circ} + \left(1-u\right)\sum_{\substack{\left(R,\alpha\right) \\ \left|\left\{i,j,k\right\}\cap R\right| =1}} p_{\left(R,\alpha\right)}^{\circ} \phi_{\widetilde{\alpha}\left(i\right)\widetilde{\alpha}\left(j\right)\widetilde{\alpha}\left(k\right)} .
\end{align}
\end{corollary}

\section{The structure-coefficient theorem}
Our main result is an expanded version of the structure-coefficient theorem of \citet{tarnita:PNAS:2011}, which provides for the computation of structure coefficients.
\begin{theorem}\label{thm:mainTheorem}
For every $y\in\left\{1,\dots ,n\right\}$, the mean frequency of strategy $y$ satisfies
\begin{align}
\frac{1}{N} &\sum_{i=1}^{N}\mathbb{P}_{\MSS}\left[ X_{i}=y\right] \nonumber \\
&= \frac{1}{n} + \delta\frac{1}{n}\left(\Lambda_{1}\left(a_{yy}-\overline{a_{\ast\ast}}\right) + \Lambda_{2}\left(\overline{a_{y\ast}}-\overline{a_{\ast y}}\right) + \Lambda_{3}\left(\overline{a_{y\ast}}-\overline{a}\right)\right) + O\left(\delta^{2}\right) , \label{eq:mainTheorem}
\end{align}
where
\begin{subequations}
\begin{align}
\overline{a_{y\ast}} &= \frac{1}{n}\sum_{z=1}^{n}a_{yz} ; \\
\overline{a_{\ast y}} &= \frac{1}{n}\sum_{z=1}^{n}a_{zy}; \\
\overline{a_{\ast\ast}} &= \frac{1}{n}\sum_{z=1}^{n}a_{zz} ; \\
\overline{a} &= \frac{1}{n^{2}}\sum_{z,w=1}^{n}a_{zw} ;
\end{align}
\end{subequations}
and
\begin{subequations}\label{eq:lambda_expressions}
\begin{align}
\Lambda_{1} &= \frac{1}{u}\sum_{i,j,k=1}^{N} \pi_{i}^{\textrm{mut}} m_{k}^{ji} \sum_{\ell =1}^{N} \Omega_{k\ell} \left(\frac{\substack{-n^{2}\phi_{ik\ell} + n^{2}\left(1-u\right)\phi_{jk\ell} + n\phi_{ik} + n\phi_{i\ell} \\ - n\left(1-u\right)\phi_{jk} - n\left(1-u\right)\phi_{j\ell} + nu\phi_{k\ell} - 2u}}{\left(n-1\right)\left(n-2\right)}\right) ; \\
\Lambda_{2} &= \frac{1}{u}\sum_{i,j,k=1}^{N} \pi_{i}^{\textrm{mut}} m_{k}^{ji} \sum_{\ell =1}^{N} \Omega_{k\ell} \left(\frac{\substack{-n^{2}\phi_{ik\ell} + n^{2}\left(1-u\right)\phi_{jk\ell} + n\phi_{ik} + n\left(n-1\right)\phi_{i\ell} \\ - n\left(1-u\right)\phi_{jk} - n\left(n-1\right)\left(1-u\right)\phi_{j\ell} + nu\phi_{k\ell} - nu}}{\left(n-1\right)\left(n-2\right)}\right) ; \\
\Lambda_{3} &= \frac{1}{u}\sum_{i,j,k=1}^{N} \pi_{i}^{\textrm{mut}} m_{k}^{ji} \sum_{\ell =1}^{N} \Omega_{k\ell} \left(\frac{\substack{2n^{2}\phi_{ik\ell} - 2n^{2}\left(1-u\right)\phi_{jk\ell} - n^{2}\phi_{ik} - n^{2}\phi_{i\ell} \\ + n^{2}\left(1-u\right)\phi_{jk} + n^{2}\left(1-u\right)\phi_{j\ell} - 2nu\phi_{k\ell} + 2nu}}{\left(n-1\right)\left(n-2\right)}\right) .
\end{align}
\end{subequations}
Moreover, $\Lambda_{1}$, $\Lambda_{2}$ and $\Lambda_{3}$ are independent of the number of strategies, $n$.
\end{theorem}
\begin{remark}
The fact that \eq{mainTheorem} holds for \emph{some} $\Lambda_{1}$, $\Lambda_{2}$, and $\Lambda_{3}$ (which are independent of $n$) is the main result of \citep{tarnita:PNAS:2011}. Our contribution is the calculation of these coefficients (\eq{lambda_expressions}).
\end{remark}
\begin{proof}[Proof of \thm{mainTheorem}]
By \eq{mainBeforeJoint}, together with the definitions of $\phi_{ij}$ and $\psi_{ijk}^{\left(L\right)}$, we see that
\begin{align}
\frac{d}{d\delta}\Bigg\vert_{\delta =0}\frac{1}{N}\sum_{i=1}^{N}\mathbb{P}_{\MSS}\left[ X_{i}=y\right] &= \frac{1}{u}\sum_{i,j,k=1}^{N} \pi_{i}^{\textrm{mut}} m_{k}^{ji} \sum_{\ell =1}^{N} \Omega_{k\ell} \chi_{1} a_{yy} \nonumber \\
&\qquad +\frac{1}{u}\sum_{i,j,k=1}^{N} \pi_{i}^{\textrm{mut}} m_{k}^{ji} \sum_{\ell =1}^{N} \Omega_{k\ell} \chi_{2} \sum_{y_{\ell}\neq y} a_{yy_{\ell}} \nonumber \\
&\qquad +\frac{1}{u}\sum_{i,j,k=1}^{N} \pi_{i}^{\textrm{mut}} m_{k}^{ji} \sum_{\ell =1}^{N} \Omega_{k\ell} \chi_{3} \sum_{y_{k}\neq y} a_{y_{k}y} \nonumber \\
&\qquad +\frac{1}{u}\sum_{i,j,k=1}^{N} \pi_{i}^{\textrm{mut}} m_{k}^{ji} \sum_{\ell =1}^{N} \Omega_{k\ell} \chi_{4} \sum_{y_{k}\neq y} a_{y_{k}y_{k}} \nonumber \\
&\qquad +\frac{1}{u}\sum_{i,j,k=1}^{N} \pi_{i}^{\textrm{mut}} m_{k}^{ji} \sum_{\ell =1}^{N} \Omega_{k\ell} \chi_{5} \sum_{y_{k}\neq y}\sum_{y_{\ell}\neq y,y_{k}} a_{y_{k}y_{\ell}} , \label{eq:expdelprime}
\end{align}
where
\begin{align}
\chi_{L} &= 
\begin{cases}
\displaystyle \left(1-u\right)\frac{\psi_{jk\ell}^{\left(1\right)}}{n} - \frac{\psi_{ik\ell}^{\left(1\right)}}{n} + \frac{u}{n} \frac{\phi_{k\ell}}{n} & \displaystyle L=1 , \\
& \\
\displaystyle \left(1-u\right)\frac{\psi_{jk\ell}^{\left(2\right)}}{n\left(n-1\right)} - \frac{\psi_{ik\ell}^{\left(2\right)}}{n\left(n-1\right)} + \frac{u}{n} \frac{1-\phi_{k\ell}}{n\left(n-1\right)} & \displaystyle L=2 , \\
& \\
\displaystyle \left(1-u\right)\frac{\psi_{jk\ell}^{\left(3\right)}}{n\left(n-1\right)} - \frac{\psi_{ik\ell}^{\left(3\right)}}{n\left(n-1\right)} + \frac{u}{n} \frac{1-\phi_{k\ell}}{n\left(n-1\right)} & \displaystyle L=3 , \\
& \\
\displaystyle \left(1-u\right)\frac{\psi_{jk\ell}^{\left(4\right)}}{n\left(n-1\right)} - \frac{\psi_{ik\ell}^{\left(4\right)}}{n\left(n-1\right)} + \frac{u}{n} \frac{\phi_{k\ell}}{n} & \displaystyle L=4 , \\
& \\
\displaystyle \left(1-u\right)\frac{\psi_{jk\ell}^{\left(5\right)}}{n\left(n-1\right)\left(n-2\right)} - \frac{\psi_{ik\ell}^{\left(5\right)}}{n\left(n-1\right)\left(n-2\right)} + \frac{u}{n} \frac{1-\phi_{k\ell}}{n\left(n-1\right)} & \displaystyle L=5 .
\end{cases}
\end{align}
We first want to express \eq{expdelprime} in terms of $a_{yy}$, $\overline{a_{y\ast}}$, $\overline{a_{\ast y}}$, $\overline{a_{\ast\ast}}$, and $\overline{a}$. For fixed $\alpha$, $\beta$, $\gamma$, $\delta$, and $\epsilon$,
\begin{align}
\alpha &a_{yy} + \beta \overline{a_{y\ast}} + \gamma \overline{a_{\ast y}} + \delta \overline{a_{\ast\ast}} + \varepsilon \overline{a} \nonumber \\
&= \left(\alpha + \frac{1}{n} \beta + \frac{1}{n} \gamma + \frac{1}{n} \delta + \frac{1}{n^{2}} \varepsilon\right) a_{yy} \nonumber \\
&\qquad + \left(\frac{1}{n} \beta + \frac{1}{n^{2}} \varepsilon\right) \sum_{y_{\ell}\neq y} a_{yy_{\ell}} + \left(\frac{1}{n}\gamma + \frac{1}{n^{2}} \varepsilon\right) \sum_{y_{k}\neq y} a_{y_{k}y} \nonumber \\
&\qquad + \left(\frac{1}{n} \delta + \frac{1}{n^{2}} \varepsilon\right) \sum_{y_{k}\neq y} a_{y_{k}y_{k}} + \left(\frac{1}{n^{2}} \varepsilon\right) \sum_{y_{k}\neq y}\sum_{y_{\ell}\neq y,y_{k}} a_{y_{k}y_{\ell}} .
\end{align}
If this expression is to be equal to \eq{expdelprime}, then we must have
\begin{subequations}
\begin{align}
\alpha &= \frac{1}{u}\sum_{i,j,k=1}^{N} \pi_{i}^{\textrm{mut}} m_{k}^{ji} \sum_{\ell =1}^{N} \Omega_{k\ell} \left(\chi_{1} - \chi_{2} - \chi_{3} - \chi_{4} + 2\chi_{5}\right) ; \\
\beta &= n\frac{1}{u}\sum_{i,j,k=1}^{N} \pi_{i}^{\textrm{mut}} m_{k}^{ji} \sum_{\ell =1}^{N} \Omega_{k\ell} \left(\chi_{2}-\chi_{5}\right) ; \\
\gamma &= n\frac{1}{u}\sum_{i,j,k=1}^{N} \pi_{i}^{\textrm{mut}} m_{k}^{ji} \sum_{\ell =1}^{N} \Omega_{k\ell} \left(\chi_{3}-\chi_{5}\right) ; \\
\delta &= n\frac{1}{u}\sum_{i,j,k=1}^{N} \pi_{i}^{\textrm{mut}} m_{k}^{ji} \sum_{\ell =1}^{N} \Omega_{k\ell} \left(\chi_{4}-\chi_{5}\right) ; \\
\varepsilon &= n^{2}\frac{1}{u}\sum_{i,j,k=1}^{N} \pi_{i}^{\textrm{mut}} m_{k}^{ji} \sum_{\ell =1}^{N} \Omega_{k\ell} \chi_{5} .
\end{align}
\end{subequations}
Using the expressions for $\psi_{ijk}^{\left(L\right)}$ ($L=2,3,4,5$) in terms of $\phi_{ijk}=\psi_{ijk}^{\left(1\right)}$ (\lem{psi_ijk}), we obtain
\begin{subequations}
\begin{align}
\alpha &= \frac{1}{u}\sum_{i,j,k=1}^{N} \pi_{i}^{\textrm{mut}} m_{k}^{ji} \sum_{\ell =1}^{N} \Omega_{k\ell} \left(\frac{\substack{-n^{2}\phi_{ik\ell} + n^{2}\left(1-u\right)\phi_{jk\ell} + n\phi_{ik} + n\phi_{i\ell} \\ - n\left(1-u\right)\phi_{jk} - n\left(1-u\right)\phi_{j\ell} + nu\phi_{k\ell} - 2u}}{n\left(n-1\right)\left(n-2\right)}\right) ; \\
\beta &= \frac{1}{u}\sum_{i,j,k=1}^{N} \pi_{i}^{\textrm{mut}} m_{k}^{ji} \sum_{\ell =1}^{N} \Omega_{k\ell} \left(\frac{\substack{n\phi_{ik\ell} - n\left(1-u\right)\phi_{jk\ell} - \left(n-1\right)\phi_{ik} - \phi_{i\ell} \\ + \left(n-1\right)\left(1-u\right)\phi_{jk} + \left(1-u\right)\phi_{j\ell} - u\phi_{k\ell} + u}}{\left(n-1\right)\left(n-2\right)}\right) ; \\
\gamma &= \frac{1}{u}\sum_{i,j,k=1}^{N} \pi_{i}^{\textrm{mut}} m_{k}^{ji} \sum_{\ell =1}^{N} \Omega_{k\ell} \left(\frac{\substack{n\phi_{ik\ell} - n\left(1-u\right)\phi_{jk\ell} - \phi_{ik} - \left(n-1\right)\phi_{i\ell} \\ + \left(1-u\right)\phi_{jk} + \left(n-1\right)\left(1-u\right)\phi_{j\ell} - u\phi_{k\ell} + u}}{\left(n-1\right)\left(n-2\right)}\right) ; \\
\delta &= \frac{1}{u}\sum_{i,j,k=1}^{N} \pi_{i}^{\textrm{mut}} m_{k}^{ji} \sum_{\ell =1}^{N} \Omega_{k\ell} \left(\frac{\substack{n^{2}\phi_{ik\ell} - n^{2}\left(1-u\right)\phi_{jk\ell} - n\phi_{ik} - n\phi_{i\ell} \\ + n\left(1-u\right)\phi_{jk} + n\left(1-u\right)\phi_{j\ell} - nu\phi_{k\ell} + 2u}}{n\left(n-1\right)\left(n-2\right)}\right) ; \\
\varepsilon &= \frac{1}{u}\sum_{i,j,k=1}^{N} \pi_{i}^{\textrm{mut}} m_{k}^{ji} \sum_{\ell =1}^{N} \Omega_{k\ell} \left(\frac{\substack{-2n\phi_{ik\ell} + 2n\left(1-u\right)\phi_{jk\ell} + n\phi_{ik} + n\phi_{i\ell} \\ - n\left(1-u\right)\phi_{jk} - n\left(1-u\right)\phi_{j\ell} + 2u\phi_{k\ell} - 2u}}{\left(n-1\right)\left(n-2\right)}\right) .
\end{align}
\end{subequations}
From these expressions, we immediately see that $\alpha +\delta =\beta +\gamma +\varepsilon =0$, consistent with \citep{tarnita:PNAS:2011}. Following \citet{tarnita:PNAS:2011}, we let $\lambda_{1}\coloneqq\alpha$, $\lambda_{2}\coloneqq -\gamma$, and $\lambda_{3}\coloneqq -\varepsilon$, which allows one to write
\begin{align}
\frac{d}{d\delta}\Bigg\vert_{\delta =0}\frac{1}{N}\sum_{i=1}^{N}\mathbb{P}_{\MSS}\left[ X_{i}=y\right] &= \alpha a_{yy} + \beta \overline{a_{y\ast}} + \gamma \overline{a_{\ast y}} + \delta \overline{a_{\ast\ast}} + \varepsilon \overline{a} \nonumber \\
&= \lambda_{1}\left(a_{yy}-\overline{a_{\ast\ast}}\right) + \lambda_{2}\left(\overline{a_{y\ast}}-\overline{a_{\ast y}}\right) + \lambda_{3}\left(\overline{a_{y\ast}}-\overline{a}\right) . \label{eq:lambdaEquation}
\end{align}

In what ways do $\lambda_{1}$, $\lambda_{2}$, and $\lambda_{3}$ depend on the number of strategies, $n$? By \prop{phi_ij}, we see that $n\phi_{ij}$ is a linear function of $n$ for every $i$ and $j$. Similarly, by \prop{phi_ijk}, $n^{2}\phi_{ijk}$ is a quadratic function of $n$ for every $i$, $j$, and $k$. Turning to the terms depending on $n$ in \eq{lambda_expressions}, let
\begin{subequations}
\begin{align}
f_{1}^{ijk\ell}\left(n\right) &\coloneqq -n^{2}\phi_{ik\ell} + n^{2}\left(1-u\right)\phi_{jk\ell} + n\phi_{ik} + n\phi_{i\ell} \nonumber \\ 
&\qquad - n\left(1-u\right)\phi_{jk} - n\left(1-u\right)\phi_{j\ell} + nu\phi_{k\ell} - 2u ; \\
f_{2}^{ijk\ell}\left(n\right) &\coloneqq -n^{2}\phi_{ik\ell} + n^{2}\left(1-u\right)\phi_{jk\ell} + n\phi_{ik} + n\left(n-1\right)\phi_{i\ell} \nonumber \\ 
&\qquad - n\left(1-u\right)\phi_{jk} - n\left(n-1\right)\left(1-u\right)\phi_{j\ell} + nu\phi_{k\ell} - nu ; \\
f_{3}^{ijk\ell}\left(n\right) &\coloneqq 2n^{2}\phi_{ik\ell} - 2n^{2}\left(1-u\right)\phi_{jk\ell} - n^{2}\phi_{ik} - n^{2}\phi_{i\ell} \nonumber \\
&\qquad + n^{2}\left(1-u\right)\phi_{jk} + n^{2}\left(1-u\right)\phi_{j\ell} - 2nu\phi_{k\ell} + 2nu
\end{align}
\end{subequations}
be the numerators of the fractions in \eq{lambda_expressions}. These functions are polynomials in $n$ of at most degree two. Since $\phi_{ijk}=\phi_{ij}=1$ for every $i$, $j$, and $k$ when $n=1$, we see immediately that $f_{1}^{ijk\ell}\left(1\right) =f_{2}^{ijk\ell}\left(1\right) =f_{3}^{ijk\ell}\left(1\right) =0$. Since $\phi_{ijk}=\left(\phi_{ij}+\phi_{ik}+\phi_{jk}-1\right) /2$ for every $i$, $j$, and $k$ when $n=2$, a straightforward calculation gives $f_{1}^{ijk\ell}\left(2\right) =f_{2}^{ijk\ell}\left(2\right) =f_{3}^{ijk\ell}\left(2\right) =0$. Therefore, there exist $g_{1}^{ijk\ell}$, $g_{2}^{ijk\ell}$, and $g_{3}^{ijk\ell}$, which are independent of $n$, such that $f_{1}^{ijk\ell}\left(n\right) =\left(n-1\right)\left(n-2\right) g_{1}^{ijk\ell}$, $f_{2}^{ijk\ell}\left(n\right) =\left(n-1\right)\left(n-2\right) g_{2}^{ijk\ell}$, and $f_{3}^{ijk\ell}\left(n\right) =\left(n-1\right)\left(n-2\right) g_{3}^{ijk\ell}$. It follows that $\Lambda_{1}\coloneqq n\lambda_{1}$, $\Lambda_{2}\coloneqq n\lambda_{2}$, and $\Lambda_{3}\coloneqq n\lambda_{3}$ are independent of $n$. This independence was  established in \citep{tarnita:PNAS:2011} using a qualitative argument. Here, we can see it directly from \eq{lambda_expressions}. We then get \eq{mainTheorem} from \eq{lambdaEquation}, as desired.
\end{proof}

\subsection{Additive games with $n$ strategies}
If a matrix game is additive, then the payoff for strategy $x$ against $y$ can be decomposed into a component coming from the player using $x$ plus a component from the player using $y$. In other words, there exist $\left\{s_{x}\right\}_{x=1}^{n}$ and $\left\{t_{y}\right\}_{y=1}^{n}$ such that $a_{xy}=s_{x}+t_{y}$ for every $x,y\in\left\{1,\dots ,n\right\}$.
\begin{example}
Suppose that strategy $x$ pays $c_{x}$ to donate $b_{x}$ to the co-player. This model is a generalization of the donation game that can cover more than two different investment levels. In this case, we have $s_{x}=-c_{x}$ and $t_{y}=b_{y}$, so that the payoff for $x$ against $y$ is $a_{xy}=-c_{x}+b_{y}$.
\end{example}
\begin{corollary}
For every $y\in\left\{1,\dots ,n\right\}$, the mean-frequency of $y$ satisfies
\begin{align}
\frac{1}{N} \sum_{i=1}^{N} \mathbb{P}_{\MSS}\left[ X_{i}=y\right] &= \frac{1}{n} + \delta\frac{1}{n}\left(K_{1} \left(s_{y}-\overline{s}\right) + K_{2} \left(t_{y}-\overline{t}\right)\right) +O\left(\delta^{2}\right) ,
\end{align}
where $\overline{s}=\left(1/n\right)\sum_{z=1}^{n}s_{z}$, $\overline{t}=\left(1/n\right)\sum_{z=1}^{n}t_{z}$, and
\begin{subequations}\label{eq:zeta_additive}
\begin{align}
K_{1} &= \frac{1}{u}\sum_{i,j,k=1}^{N} \pi_{i}^{\textrm{mut}} m_{k}^{ji} \sum_{\ell =1}^{N} \Omega_{k\ell} \left( \frac{- n\phi_{ik} + n\left(1-u\right)\phi_{jk} + u}{n-1} \right) ; \\
K_{2} &= \frac{1}{u}\sum_{i,j,k=1}^{N} \pi_{i}^{\textrm{mut}} m_{k}^{ji} \sum_{\ell =1}^{N} \Omega_{k\ell} \left( \frac{- n\phi_{i\ell} + n\left(1-u\right)\phi_{j\ell} + u}{n-1} \right) .
\end{align}
\end{subequations}
\end{corollary}
\begin{proof}
With respect to a fixed strategy, $y$, an additive game satisfies
\begin{subequations}
\begin{align}
a_{yy} &= s_{y}+t_{y} ; \\
\overline{a_{y\ast}} &= s_{y}+\overline{t} ; \\
\overline{a_{\ast y}} &= \overline{s}+t_{y} ; \\
\overline{a_{\ast\ast}} &= \overline{s}+\overline{t} ; \\
\overline{a} &= \overline{s}+\overline{t} .
\end{align}
\end{subequations}
With $\Lambda_{1}$, $\Lambda_{2}$, and $\Lambda_{3}$ the coefficients of \thm{mainTheorem}, we have
\begin{align}
\Lambda_{1} &\left(a_{yy}-\overline{a_{\ast\ast}}\right) + \Lambda_{2}\left(\overline{a_{y\ast}}-\overline{a_{\ast y}}\right) + \Lambda_{3}\left(\overline{a_{y\ast}}-\overline{a}\right) \nonumber \\
&= \left(\Lambda_{1}+\Lambda_{2}+\Lambda_{3}\right)\left(s_{y}-\overline{s}\right) + \left(\Lambda_{1}-\Lambda_{2}\right)\left(t_{y}-\overline{t}\right) .
\end{align}
Using \eq{lambda_expressions}, we see that $\Lambda_{1}+\Lambda_{2}+\Lambda_{3}=K_{1}$ and $\Lambda_{1}-\Lambda_{2}=K_{2}$ as defined in \eq{zeta_additive}.
\end{proof}

The structure-coefficient theorem for additive games (with any finite number of strategies) can thus be reduced to solving a simpler linear system of size $O\left(N^{2}\right)$ rather than $O\left(N^{3}\right)$.

\subsection{Games with two strategies}
Another case in which the structure-coefficient theorem can be simplified is when there are $n=2$ strategies. The game need not be additive, so we assume that the payoff matrix of \eq{symmetricPayoffMatrix} involves four generic parameters, $a_{11}$, $a_{12}$, $a_{21}$, and $a_{22}$. Here, the structure coefficient reduces to the main result of \citep{tarnita:JTB:2009} (recapitulated below), and our contributions are formulas for evaluating this result. Without a loss of generality, we assume that the strategy in question is $y=1$.
\begin{corollary}
In a two-strategy matrix game, the mean frequency of strategy $y=1$ satisfies
\begin{align}
\frac{1}{N}\sum_{i=1}^{N}\mathbb{P}_{\MSS}\left[ X_{i}=1\right] &= \frac{1}{2} + \delta\frac{1}{2}\left(K_{1}\left(a_{11}-a_{22}\right) +K_{2}\left(a_{12}-a_{21}\right)\right) + O\left(\delta^{2}\right) ,
\end{align}
where
\begin{subequations}\label{eq:coefficients_n2}
\begin{align}
K_{1} &= \frac{1}{2u} \sum_{i,j,k=1}^{N} \pi_{i}^{\textrm{mut}} m_{k}^{ji} \sum_{\ell =1}^{N} \Omega_{k\ell} \left( - \left(\phi_{ik}+\phi_{i\ell}\right) + \left(1-u\right)\left(\phi_{jk}+\phi_{j\ell}\right) + u \right) ; \\
K_{2} &= \frac{1}{2u} \sum_{i,j,k=1}^{N} \pi_{i}^{\textrm{mut}} m_{k}^{ji} \sum_{\ell =1}^{N} \Omega_{k\ell} \left( - \left(\phi_{ik}-\phi_{i\ell}\right) + \left(1-u\right)\left(\phi_{jk}-\phi_{j\ell}\right) \right) .
\end{align}
\end{subequations}
\end{corollary}
\begin{proof}
With respect to strategy $y=1$ in a two-strategy game,
\begin{subequations}
\begin{align}
a_{yy} &= a_{11} ; \\
\overline{a_{y\ast}} &= \frac{1}{2}\left(a_{11}+a_{12}\right) ; \\
\overline{a_{\ast y}} &= \frac{1}{2}\left(a_{11}+a_{21}\right) ; \\
\overline{a_{\ast\ast}} &= \frac{1}{2}\left(a_{11}+a_{22}\right) ; \\
\overline{a} &= \frac{1}{4}\left(a_{11}+a_{12}+a_{21}+a_{22}\right) .
\end{align}
\end{subequations}
With $\Lambda_{1}$, $\Lambda_{2}$, and $\Lambda_{3}$ the coefficients of \thm{mainTheorem}, we have
\begin{align}
\Lambda_{1} &\left(a_{yy}-\overline{a_{\ast\ast}}\right) + \Lambda_{2}\left(\overline{a_{y\ast}}-\overline{a_{\ast y}}\right) + \Lambda_{3}\left(\overline{a_{y\ast}}-\overline{a}\right) \nonumber \\
&= \frac{1}{4}\left(2\Lambda_{1}+\Lambda_{3}\right) \left(a_{11}-a_{22}\right) + \frac{1}{4}\left(2\Lambda_{2}+\Lambda_{3}\right) \left(a_{12}-a_{21}\right) .
\end{align}
Using \eq{lambda_expressions}, we have $\left(2\Lambda_{1}+\Lambda_{3}\right) /4=K_{1}$ and $\left(2\Lambda_{2}+\Lambda_{3}\right) /4=K_{2}$ as in \eq{coefficients_n2}.
\end{proof}

Therefore, even for non-additive games with two strategies, evaluating the structure-coefficient theorem can be reduced to solving a linear system with $O\left(N^{2}\right)$ (rather than $O\left(N^{3}\right)$) terms.

\section{Extension to asymmetric games} \label{sec:asymmetric}
Here, we address the generality of the structure-coefficient theorem of \citet{tarnita:JTB:2009,tarnita:PNAS:2011} and give an example of a natural model of finite-population evolution that falls outside of its original scope but is covered by the methods presented here.

Instead of \eq{symmetricPayoffMatrix} describing every interaction, there could be a separate payoff matrix for every pair. That is, the payoff matrix for $i$ against $j$ could depend on both $i$ and $j$, of the form
\begin{align}
\bordermatrix{
&\ 1 &\ 2 &\ \cdots &\ n \cr
1\,\,\, &\ a_{11}^{ij} &\ a_{12}^{ij} &\ \cdots &\ a_{1n}^{ij} \cr
2 &\ a_{21}^{ij} &\ a_{22}^{ij} &\ \cdots &\ a_{2n}^{ij} \cr
\,\vdots &\ \vdots & \vdots & \ddots &\ \vdots \cr
n &\ a_{n1}^{ij} &\ a_{n2}^{ij} &\ \cdots &\ a_{nn}^{ij}
} , \label{eq:asymmetricPayoffMatrix}
\end{align}
which is an asymmetric game \citep{maynardsmith:CUP:1982,broom:TF:2013}. Now the total payoff to individual $i$ in state $\vx$ is
\begin{align}
U_{i}\left(\vx\right) &= \sum_{\ell =1}^{N} \sum_{y_{i},y_{\ell}=1}^{n} \mathbb{1}_{y_{i}}\left(x_{i}\right)\mathbb{1}_{y_{\ell}}\left(x_{\ell}\right) a_{y_{i}y_{\ell}}^{i\ell} .
\end{align}
In contrast to \eq{Ui_symmetric}, here we suppress $\Omega_{i\ell}$ due to the dependence of $a_{y_{i}y_{\ell}}^{i\ell}$ on both $i$ and $\ell$.

We have the following structure-coefficient theorem for asymmetric matrix games, such as the games described in \citet{mcavoy:JRSI:2015}.
\begin{theorem}\label{thm:asymmetricTheorem}
For every $y\in\left\{1,\dots ,n\right\}$, the mean frequency of strategy $y$ satisfies
\begin{align}
\frac{1}{N} &\sum_{i=1}^{N}\mathbb{P}_{\MSS}\left[ X_{i}=y\right] \nonumber \\
&= \frac{1}{n} + \delta\frac{1}{n}\sum_{k,\ell =1}^{N}\left(\Lambda_{1}^{k\ell}\left(a_{yy}^{k\ell}-\overline{a_{\ast\ast}^{k\ell}}\right) + \Lambda_{2}^{k\ell}\left(\overline{a_{y\ast}^{k\ell}}-\overline{a_{\ast y}^{k\ell}}\right) + \Lambda_{3}^{k\ell}\left(\overline{a_{y\ast}^{k\ell}}-\overline{a^{k\ell}}\right)\right) + O\left(\delta^{2}\right) , \label{eq:asymmetricTheorem}
\end{align}
where
\begin{subequations}\label{eq:lambda_expressions_asymmetric}
\begin{align}
\Lambda_{1}^{k\ell} &= \frac{1}{u}\sum_{i,j=1}^{N} \pi_{i}^{\textrm{mut}} m_{k}^{ji} \left(\frac{\substack{-n^{2}\phi_{ik\ell} + n^{2}\left(1-u\right)\phi_{jk\ell} + n\phi_{ik} + n\phi_{i\ell} \\ - n\left(1-u\right)\phi_{jk} - n\left(1-u\right)\phi_{j\ell} + nu\phi_{k\ell} - 2u}}{\left(n-1\right)\left(n-2\right)}\right) ; \\
\Lambda_{2}^{k\ell} &= \frac{1}{u}\sum_{i,j=1}^{N} \pi_{i}^{\textrm{mut}} m_{k}^{ji} \left(\frac{\substack{-n^{2}\phi_{ik\ell} + n^{2}\left(1-u\right)\phi_{jk\ell} + n\phi_{ik} + n\left(n-1\right)\phi_{i\ell} \\ - n\left(1-u\right)\phi_{jk} - n\left(n-1\right)\left(1-u\right)\phi_{j\ell} + nu\phi_{k\ell} - nu}}{\left(n-1\right)\left(n-2\right)}\right) ; \\
\Lambda_{3}^{k\ell} &= \frac{1}{u}\sum_{i,j=1}^{N} \pi_{i}^{\textrm{mut}} m_{k}^{ji} \left(\frac{\substack{2n^{2}\phi_{ik\ell} - 2n^{2}\left(1-u\right)\phi_{jk\ell} - n^{2}\phi_{ik} - n^{2}\phi_{i\ell} \\ + n^{2}\left(1-u\right)\phi_{jk} + n^{2}\left(1-u\right)\phi_{j\ell} - 2nu\phi_{k\ell} + 2nu}}{\left(n-1\right)\left(n-2\right)}\right) .
\end{align}
\end{subequations}

Moreover, $\Lambda_{1}^{ij}$, $\Lambda_{2}^{ij}$ and $\Lambda_{3}^{ij}$ are independent of the number of strategies, $n$.
\end{theorem}
The proof of \thm{asymmetricTheorem} can be extracted from that of \thm{mainTheorem} by not summing over $k$ and $\ell$. In particular, no further work is needed to get the structure coefficients for asymmetric games.

\begin{example}
Let $\left(\Omega_{ij}\right)_{i,j=1}^{N}$ be a fixed matrix and suppose that $n=2$. If $a_{11}^{ij}=\Omega_{ij}a_{11}$, $a_{12}^{ij}=\Omega_{ij}a_{12}$, $a_{21}^{ij}=\Omega_{ij}a_{21}$, and $a_{22}^{ij}=\Omega_{ij}a_{22}$ for some fixed $a_{11},a_{12},a_{21},a_{22}\in\mathbb{R}$, then we define
\begin{subequations}\label{eq:compositeCoefficients}
\begin{align}
K_{1}\left(\Omega\right) &\coloneqq \sum_{k,\ell =1}^{N} \Omega_{k\ell} K_{1}^{k\ell} ; \\
K_{2}\left(\Omega\right) &\coloneqq \sum_{k,\ell =1}^{N} \Omega_{k\ell} K_{2}^{k\ell} .
\end{align}
\end{subequations}
For example, in the case that $\Omega_{ij}=w_{ij}$ (adjacency matrix), we get the two coefficients of \citet[][Equation 15]{tarnita:JTB:2009} under accumulated payoffs. When $\Omega_{ij}=w_{ij}/w_{i}$ (where $w_{i}=\sum_{k=1}^{N}w_{ik}$), we get the two coefficients of \citet[][Equation 15]{tarnita:JTB:2009} under averaged payoffs. In other words, we can easily realize structure coefficients for symmetric games as ``composite'' coefficients.
\end{example}

\section{Examples illustrating the effects of selection intensity}
To supplement the examples in the main text, we consider the effects of varying selection intensity beyond the value considered there ($\delta =0.05$). Since our results are formulated for weak selection, we expect agreement between the simulation and analytical results when $\delta$ is sufficiently small. \fig{ff_vs_pp} shows that for $\delta =0.05$, the analytical results are already in close agreement with those of the simulated process. As one might expect, reducing $\delta$ from $0.05$ to $0.01$ results in even closer agreement, which we illustrate in \fig{ff_vs_pp_weaker}. However, since there is comparatively more noise in the process when $\delta =0.01$ than when $\delta =0.05$, we see more variation in the simulation data for this smaller selection intensity when we use the same number of generations ($10^{8}$).

\begin{figure}
\centering
\includegraphics[width=0.8\textwidth]{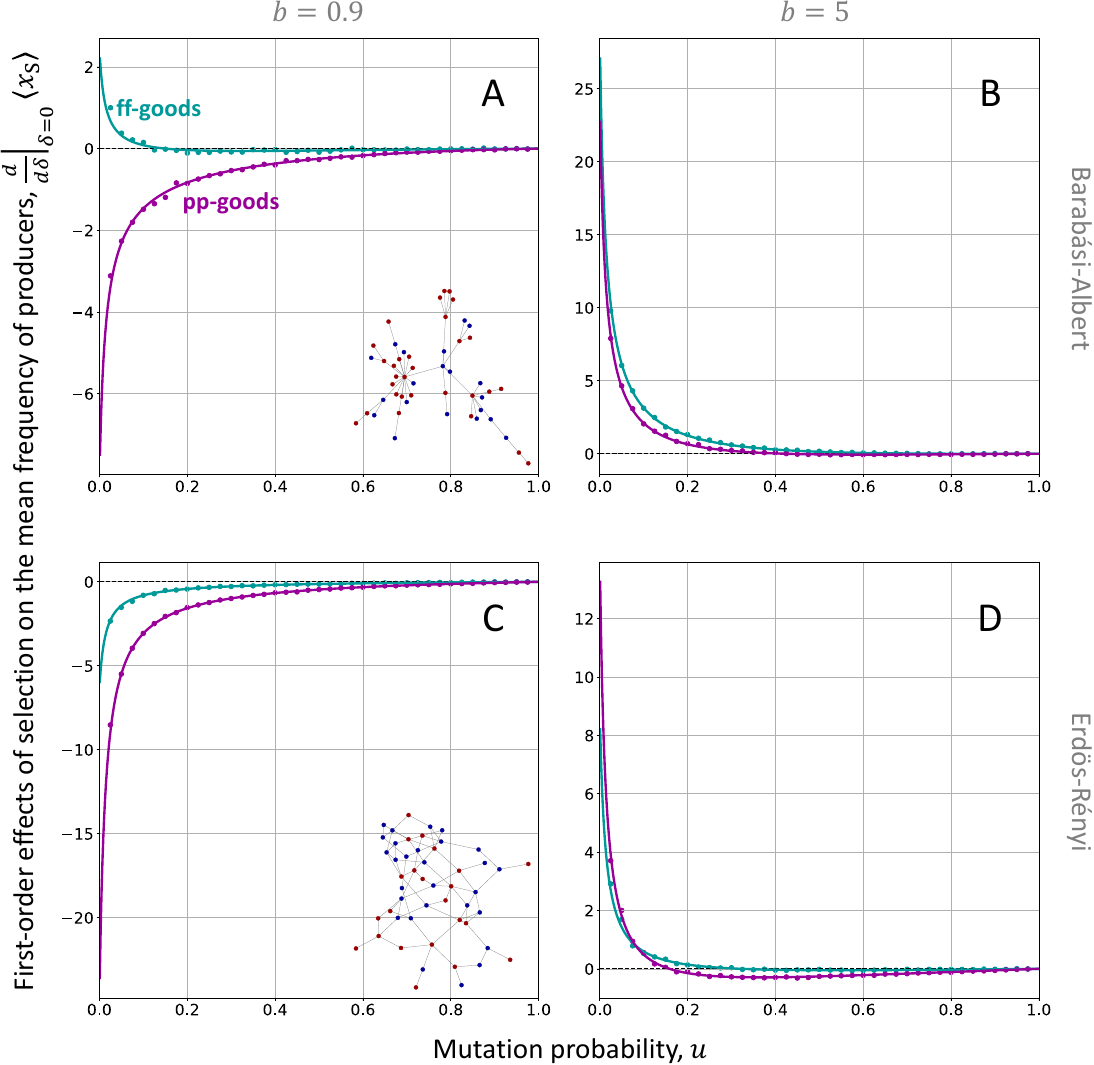}
\caption{Effects of selection on the mean frequencies of producers of ff- and pp-goods. The panels here are identical to those of \fig{ff_vs_pp} in the main text, with the exception that here simulations are performed for a selection intensity of $\delta =0.01$ rather than $\delta =0.05$. The result is that the first-order effects via simulations more closely resemble those of the calculations, but there is also more noise when the mean is taken over the same number of generations ($10^{8}$).\label{fig:ff_vs_pp_weaker}}
\end{figure}

\fig{stronger_selection} shows the effects of increasing the selection intensity, from $\delta=0.05$, which was used for the simulations in \fig{ff_vs_pp}, to $\delta=0.2$ and then to $\delta=0.5$. \fig{stronger_selection} gives a slightly different illustration of mutation-selection dynamics, showing the effects on the mean frequencies of producers directly rather than on the first-order effects of selection as in \fig{ff_vs_pp}. When the selection intensity is small, the analytical predictions of \eq{xA_expansion} agree closely with the results of simulations (panels A and D of \fig{stronger_selection}, or, equivalently, panels A and D of \fig{ff_vs_pp}, since the parameters of these respective panels are the same in both figures). Moreover, all of the qualitative conclusions based on the analytical results, which assume weak selection, continue to hold as $\delta$ increases. Quantitatively they differ, most obviously in that the magnitude of effects on the mean frequency is larger when $\delta$ is larger and more subtly in that changing $\delta$ shifts the crossover point for the direction of selection on the prosocial behavior (panels B--C and E--F of \fig{stronger_selection}). The first may not be surprising, as selection is a deterministic force and all we are doing here is changing its strength, but it is clear that conclusions based on weak mutation are less robust than those based on weak selection. The simulations suggest that the second, subtler effect depends on population structure. In panels B--C of \fig{stronger_selection}, increasing $\delta$ shifts the crossover point to smaller values of $u$.  In panels E--F of \fig{stronger_selection}, increasing $\delta$ shifts the crossover point to larges values of $u$.

\begin{figure}
\centering
\includegraphics[width=\textwidth]{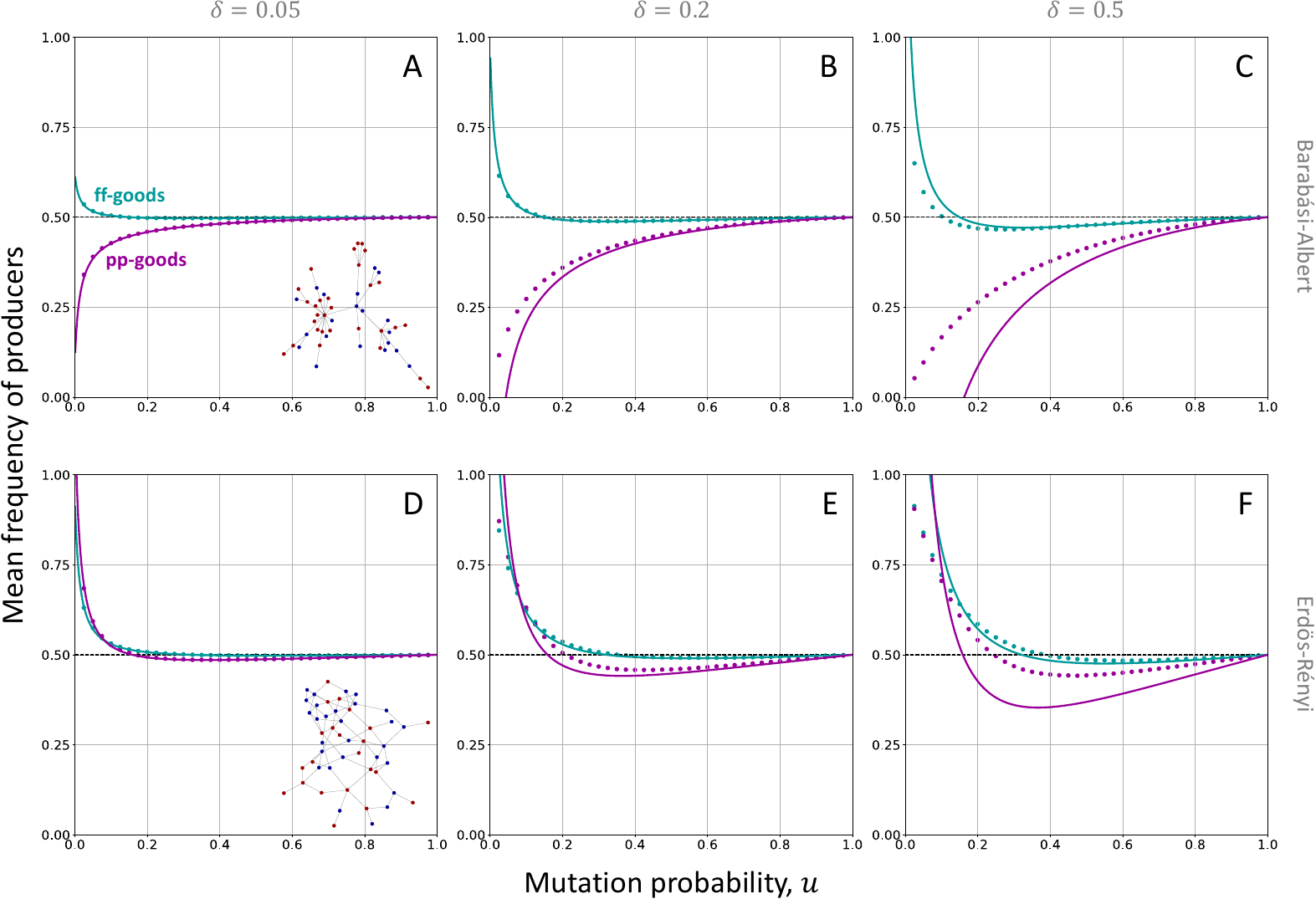}
\caption{Mean frequencies of simple prosocial behaviors as the selection intensity varies. Panels A--C depict results for the Barab\'{a}si-Albert graph and payoff values of \fig{ff_vs_pp}A ($b=0.9$ and $c=1$), while panels D--F show producer frequencies for the Erd\"{o}s-R\'{e}nyi graph and payoff values of \fig{ff_vs_pp}D ($b=5$ and $c=1$). Simulation results (dots) shown are obtained by averaging the mean frequencies of producers over $10^{8}$ generations. The solid lines depict the quantity $1/2+\delta\frac{d}{d\delta}\Big\vert_{\delta =0}\left<x_{\textrm{S}}\right>$, ignoring higher-order terms in $\delta$. As expected, these higher-order terms matter more for larger $\delta$, where the first-order calculations begin to differ from simulation results. However, they still capture many of the qualitative features of the effects of mutation probabilities on mean frequencies, e.g. whether a social good is favored for small versus large $u$.\label{fig:stronger_selection}}
\end{figure}

Finally, in \fig{stronger_selection_same_plot}, we depict the results of \fig{stronger_selection} in a slightly different way, with different selection intensities shown in the same panel. We do not include the predictions of weak selection in this panel because \fig{stronger_selection} illustrates that the exact calculations do not always closely match up with the simulations when $\delta$ is larger. These examples reveal an effect akin to amplifying selection: When the mean frequency of producers is larger than $1/2$, increasing the strength of selection further increases this frequency. When it is below $1/2$, increasing the strength of selection further decreases this frequency. This finding does not strictly hold for all points (i.e. near the crossover point mentioned above), but it is a broad trend supported by all four panels. We leave this as a topic for future investigation.

\begin{figure}
\centering
\includegraphics[width=0.8\textwidth]{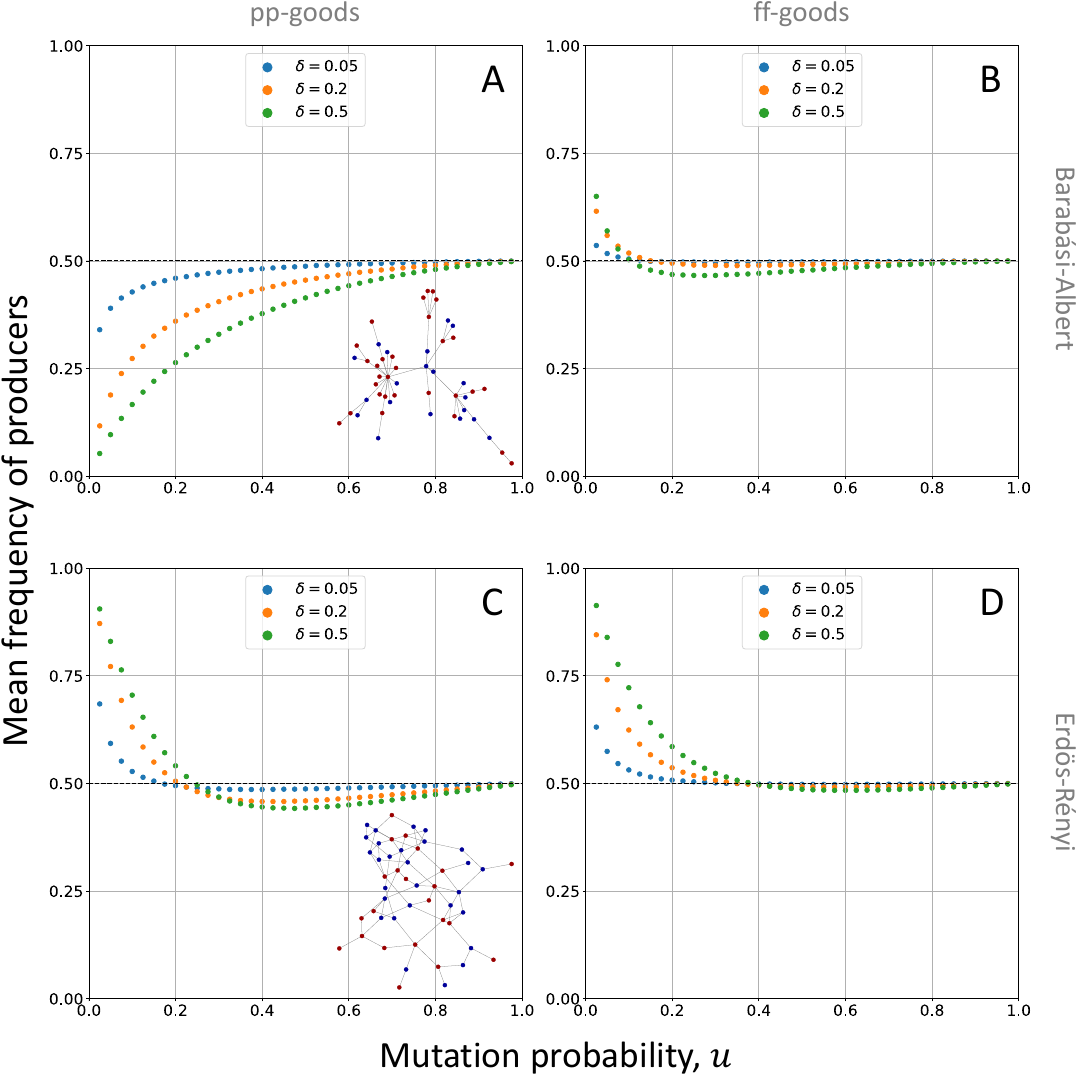}
\caption{Mean frequencies of simple prosocial behaviors, with different selection intensities depicted on the same panel. All parameters are the same as in \fig{stronger_selection}. The comparison here is of simulation results as the selection intensity varies. All panels seem to indicate (roughly) an amplification effect of selection, with stronger selection increasing (decreasing) the advantage of producers who are favored (disfavored) under weak selection.\label{fig:stronger_selection_same_plot}}
\end{figure}

\end{document}